\documentclass[12pt]{article}

\usepackage{amscd}     
\usepackage{amsmath}   
\usepackage{amssymb}   
\usepackage{amsthm}    
\usepackage{hyperref}  
\usepackage{bookmark}  
\usepackage{mathtools} 
\usepackage{epstopdf}  
\usepackage{verbatim}  
\usepackage{graphicx}  
\usepackage{color}     
\usepackage{xcolor}
\usepackage{soul}
\usepackage[skip=2pt,font={small, it}]{caption} 
\usepackage{algorithm} 
\usepackage{algorithmicx}
\usepackage{algpseudocode}
\usepackage[author={Barak Sober}]{pdfcomment}
\usepackage{geometry}
\graphicspath{{./}}

\newtheorem{definition}{Definition}

\newtheorem{Theorem}{Theorem}[section]
\newtheorem{Lemma}[Theorem]{Lemma}

\newtheorem{corollary}[Theorem]{Corollary}
\newtheorem{assumption}[Theorem]{Assumption}
\newtheorem{prop}[Theorem]{Proposition}
\theoremstyle{remark}
\newtheorem{remark}[Theorem]{Remark}

\newcommand{\RR}{\mathbb{R}}

\newcommand{\MM}{\mathcal{M}}
\newcommand{\NN}{\mathcal{N}}
\newcommand{\UU}{\mathcal{U}}
\newcommand{\OO}{\mathcal{O}}

\newcommand{\GG}{\mathcal{G}}
\newcommand{\hrho}{h\text{-}\rho\text{-}\delta}

\newcommand{\defeq}{\stackrel{\textrm{def}}{=}}
\newcommand{\eqdef}{=\vcentcolon}
\newcommand{\abs}[1]{\left\vert #1 \right\vert}
\newcommand{\norm}[1]{\left\Vert #1 \right\Vert}
\newcommand{\pdev}[2]{\frac{\partial #1}{\partial #2}}

\newcommand{\rpm}{\raisebox{.2ex}{$\scriptstyle\pm$}}
\newcommand{\eps}{\varepsilon}
\DeclareMathOperator*{\argmin}{\arg\!\min}

\DTMlangsetup[en-GB]{ord=raise,monthyearsep={,\space}}

\begin{document}
\title{Manifold Approximation by Moving Least-Squares Projection (MMLS)}
\author{Barak Sober~~David Levin \\
\small{School of Mathematical Sciences, Tel Aviv University, Israel}
}

\maketitle
\begin{abstract}
In order to avoid the curse of dimensionality, frequently encountered in Big Data analysis, there was vast development in the field of linear and nonlinear dimension reduction techniques in recent years. These techniques (sometimes referred to as manifold learning) assume that the scattered input data is lying on a lower-dimensional manifold, thus the high dimensionality problem can be overcome by learning the lower dimensionality behavior. However, in real-life applications, data is often very noisy. In this work, we propose a method to approximate $\MM$ a $d$-dimensional $C^{m+1}$ smooth submanifold of $\RR^n$ ($d \ll n$) based upon noisy scattered data points (i.e., a data cloud). We assume that the data points are located ``near" the lower-dimensional manifold and suggest a non-linear moving least-squares projection on an approximating $d$-dimensional manifold. Under some mild assumptions, the resulting approximant is shown to be infinitely smooth and of high approximation order (i.e., $\OO(h^{m+1})$, where $h$ is the fill distance and $m$ is the degree of the local polynomial approximation). The method presented here assumes no analytic knowledge of the approximated manifold and the approximation algorithm is linear in the large dimension $n$. Furthermore, the approximating manifold can serve as a framework to perform operations directly on the high dimensional data in a computationally efficient manner. This way, the preparatory step of dimension reduction, which induces distortions to the data, can be avoided altogether.
\end{abstract}

\noindent\textbf{keywords:} Manifold learning, Manifold approximation, Moving Least-Squares, Dimension reduction, Manifold denoising

(Numerical analysis - Numerical approximation and computational geometry)

\section{Introduction}

The digital revolution in which we live has resulted in vast amounts of high dimensional data. This proliferation of knowledge inspires both the industrial and research communities to explore the underlying patterns of these information-seas. However, navigating through these resources encompasses both computational and statistical difficulties. Whereas the computational challenge is clear when dealing with Big-Data, the statistical issue is a bit more subtle. 

Apparently, data lying in very high dimensions is usually sparsely distributed - a phenomenon sometimes referred to by the name \textit{the curse of dimensionality}. Explicitly, one million data points, arbitrarily distributed in $\mathbb{R}^{100}$ is too small a data-set for data analysis. Therefore, the effectiveness of pattern recognition tools is somewhat questionable, when dealing with high dimensional data \cite{hughes1968mean, donoho2000highDimensions, bellman1957dynamic}.
However, if these million data points are assumed to be situated near a low dimensional manifold, e.g., up to six dimensions, then, in theory, we have enough data points for valuable data analysis.

One way to overcome the aforementioned obstacle is to assume that the data points are situated on a lower-dimensional manifold and apply various algorithms to learn the underlying manifold, prior to applying other analyses. In most \textit{manifold learning} algorithms, the process of learning a manifold from a point-cloud is, in fact, the process of embedding the point-cloud into a lower-dimensional Euclidean space. These procedures are sometimes called dimension reduction, which is a more suitable name. 

Perhaps the most well-known dimension reduction technique, presupposing that the data originates from a linear manifold, is the Principal Component Analysis (PCA)\cite{jolliffe2002PCA}. The PCA solves the problem of finding a projection on a linear sub-space preserving as much as possible of the data's variance. Yet, in case the relationships between the scattered data points are more complicated than that, there is no clear-cut solution. The methods used in dimension reduction can range between \cite{lee2007nonlinear}: linear or non-linear; have a continuous or discrete model; perform implicit or explicit mappings. Furthermore, the type of criterion each method tries to optimize may be completely different. For example: \textit{multidimensional scaling} methods \cite{torgerson1952MDS}, \textit{curvilinear component analysis} \cite{demartines1997curvilinear} and \textit{Isomap} \cite{tenenbaum2000isomap} aim at preserving distances (either Euclidean or geodesic, local or global) between the data points; \textit{Kernel PCA} methods aim at linearization of the manifold through using a kernel function in the scalar product \cite{scholkopf1998KPCA}; \textit{Self Organizing Maps} (SOM) aims at fitting a $d$-dimensional grid to the scattered data through minimizing distances to some prototypes \cite{von1973SOM, kohonen1982SOM, kohonen2001SOMbook, lee2007nonlinear}; \textit{General Topographic Mapping} fits a grid to the scattered data as well, through maximization of likelihood approximation \cite{bishop1996gtm, lee2007nonlinear}; \textit{Local Linear Embedding} (LLE) aims at maintaining angles between neighboring points \cite{roweis2000LLE, saul2003LLE}; \textit{Laplacian Eigenmaps} approximate an underlying manifold through eigenfunctions of the Graph Laplacian \cite{belkin2003laplacian}; \textit{Diffusion maps} use the modeling of diffusion processes and utilize Markov Chain techniques to find representation of meaningful structures \cite{coifman2006diffusion}; and \textit{Maximum Variance Unfolding} uses semi-definite programming techniques to maximize the variance of non-neighboring points \cite{weinberger2006MaximumVarianceUnfolding}. 

It is interesting to note that all of the aforementioned dimension reduction techniques aim at finding a global embedding of the data into $\RR^{\tilde{d}}$ ($\tilde{d} > d$) in a ``nearly" isometric fashion. Theoretically, a closed manifold $\MM$ of dimension $d$ can be isometrically embedded by a $C^1$ mapping into $\RR^{2d}$ due to Nash's theorem \cite{nash1954c1imbedding}. However, it is not clear how to construct such an embedding when there are merely discrete samples of the manifold, without any knowledge regarding the Riemannian metric or coordinate charts. 

Furthermore, albeit the proliferation of methods performing dimension reduction, less attention has been aimed at denoising or approximating an underlying manifold from scattered data. This pre-processing denoising step could be crucial, especially when the dimension reduction technique being utilized relies upon differential operators (e.g., eigenfunctions of the graph Laplacian). For clean samples of a manifold, a simplicial reconstruction has been suggested as early as 2002 by Freedman \cite{freedman2002efficient}. Another simplicial manifold reconstruction is presented in \cite{cheng2005manifold}, but the algorithm depends exponentially on the dimension. An elaboration and development of Freedman's method, utilizing tangential Delaunay complexes, is presented in \cite{boissonnat2014manifold}. In the latter, the algorithm is claimed to be linear in the ambient dimension. For the case of noisy samples of a manifold, there were works aiming at manifold denoising. A statistical approach relying upon graph-based diffusion process is presented in \cite{hein2006manifoldDenoising}. Another work dealing with a locally linear approximation of the manifold is presented in \cite{gong2010LLD}. 

In our work, we assume that our high dimensional data (in $\mathbb{R}^n$) lies near (or on) a low dimensional smooth manifold (or manifolds), of a known dimension $d$, with no boundary. We aim at approximating the manifold, handling noisy data, and understanding the local structure of the manifold. Our approach naturally leads to measuring distances from the manifold and to approximating functions defined over the manifold \cite{sober2017approximation}.

The main tool we use for approximating a $C^{m+1}$ smooth manifold is a non-linear Moving Least-Squares approach, generalizing the surface approximating algorithm presented in \cite{levin2004mesh}. The approximation we derive below, is based upon a local projection procedure which results in a $C^\infty$ smooth $d$-dimensional manifold (Theorem \ref{thm:ManifoldMMLS}) of approximation order $\OO(h^{m+1})$, where $h$ is the fill distance of the data cloud (Theorem \ref{thm:OrderMMLS}). Furthermore, the suggested implementation for this projection procedure is of complexity order  $\OO(n)$ (neglecting the dependency upon the lower dimension $d$). The general idea behind this projection follows from the definition of a differentiable manifold using coordinate charts, collected in a mathematical atlas. The proposed mechanism takes this concept to its fullest extent and involves the construction of a different local coordinate chart for each point on the manifold.

It is worth noting that throughout the article, we use the term smooth manifold to address a submanifold in $\RR^n$, which is smooth with respect to the smoothness structure of $\RR^n$. Explicitly, if a manifold can be considered locally as a smooth graph of a function, it is said to be smooth.

In Section \ref{sec:preliminaries}, we start the presentation by reviewing the method of moving least-squares for multivariate scattered data function approximation \cite{mclain1974drawing}, and its adaptation to the approximation of surfaces from a cloud of points\cite{levin2004mesh}. In Section \ref{sec:ManifoldMLS} we present the generalization of the projection method of \cite{levin2004mesh} to the general case of approximating a $d$-dimensional submanifold in $\mathbb{R}^n$. In Section \ref{sec:Theory} we discuss the topological dimension, the smoothness properties as well as the approximation power of the approximating manifold. We conclude by several numerical examples in Section \ref{sec:examples}.

\section{Preliminaries}
\label{sec:preliminaries}
As mentioned above, the Moving Least-Squares (MLS) method was originally designed for the purpose of smoothing and interpolating scattered data, sampled from some multivariate function \cite{mclain1974drawing, lancaster1981surfaces, nealen2004short}. 
The general idea was to utilize the Least-Squares mechanism on a local rather than a global level. 
This way, one can regress and capture the local trends of the data and better reconstruct a wider set of functions, than those described by mere polynomials.
Later, the MLS mechanism evolved to deal with the more general case of surfaces, which can be viewed as a function locally rather than globally \cite{levin2004mesh, levin1998MLSapproximation}. Accordingly, in this brief overview of the topic we shall follow the rationale of \cite{levin2004mesh} and start by presenting the problem of function approximation, continue with surface approximation and in section \ref{sec:ManifoldMLS} we generalize the MLS projection procedure for a Riemannian submanifold of $\RR^n$. 

We would like to stress upfront that throughout the article $\norm{\cdot}$ represents the standard Euclidean norm.

\subsection{MLS for function approximation}
Let $\lbrace x_i \rbrace_{i=1}^{I}$ be a set of distinct scattered points in $\mathbb{R}^d$ and let $\lbrace f(x_i) \rbrace_{i=1}^{I}$ be the corresponding sampled values of some function $f:\mathbb{R}^d \rightarrow \mathbb{R}$. Then, the $m^{th}$ degree moving least-squares approximation to $f$ at a point $ x \in \mathbb{R}^d $ is defined as $p_x(x)$ where:
\begin{equation}
\label{eq:basicMLS}
p_x = \argmin_{p \in \Pi_m^d} \sum_{i=1}^{I} (p(x_i) - f(x_i))^2 \theta(\| x - x_i\|)
,\end{equation}
$ \theta(s) $ is a non-negative weight function (rapidly decreasing as $s \rightarrow \infty$), $\| \cdot \|$ is the Euclidean norm and $\Pi_m^d$ is the space of polynomials of total degree $m$ in $\mathbb{R}^d$. We then define the MLS approximation of the function to be 
\begin{equation}
    \tilde{p}(x)\defeq p_x(x) \approx f(x)
\label{eq:MLSfunctionApproximant}\end{equation}
Notice, that the rapid decay of $\theta$ makes the polynomial approximation fit the data points on locally, and so $p_x$ would change in order to fit the local behavior of the data.
Furthermore, if $\theta(s)$ is of finite support then the approximation is made local, and if $\theta(0)=\infty$ the MLS approximation interpolates the data.

We wish to quote here two previous results regarding the approximation, presented in \cite{levin1998MLSapproximation}. In section \ref{sec:ManifoldMLS} we will prove properties extending these theorems to the general case of a $d$-dimensional Riemannian manifold residing in $\RR^n$.
\begin{Theorem}
Let $\theta(t)$ be a weight function such that $\lim_{t\rightarrow 0}\theta(t) = \infty$ and $\theta \in C^\infty$ at $t\neq 0$ (i.e., the scheme is interpolatory), and let the distribution of the data points $\lbrace x_i \rbrace_{i=1}^{I}$ be such that the problem is well conditioned (i.e., the least-squares matrix is invertible). Then the MLS approximation is a $C^\infty$ function interpolating the data points $\lbrace f(x_i) \rbrace_{i=1}^{I}$.
.\label{thm:SmoothMLSfunctions}\end{Theorem}
The second result, dealing with the approximation order, necessitates the introduction of the following definition:
\begin{definition}
\textbf{$\hrho$ sets of fill distance $h$, density $\leq \rho$, and separation $\geq \delta$.}
Let $\Omega$ be a domain in $\RR^d$, and consider sets of data points in $\Omega$. We say that the set $X = \lbrace x_i \rbrace_{i=1}^I$ is an  $\hrho$ set if:
\begin{enumerate}
\item $h$ is the fill distance with respect to the domain $\Omega$
\begin{equation}
h = \sup_{x\in\Omega} \min_{x_i \in X} \norm{x - x_i}
\label{def:h},
\end{equation}

\item 
\begin{equation}
\#\left\lbrace X \cap \overline{B}_{qh}(y)  \right\rbrace \leq \rho \cdot q^d, ~~ q\geq 1, ~~ y \in \RR^d.
\label{def:rho}\end{equation}
Here $\# Y$ denotes the number of elements in a given set $Y$, while $\overline{B}_r(x)$ is the closed ball of radius $r$ around $x$.
\item $\exists \delta>0$ such that
\begin{equation}
\norm{x_i - x_j} \geq h \delta, ~~ 1 \leq i < j \leq I
\label{def:delta}\end{equation}
\end{enumerate}
\label{def:h-rho-delta}\end{definition}

\begin{remark}
Notice that albeit its name the ``fill distance" is not a metric nor a distance defined between objects. This name refers to the maximal gap in the data.
\end{remark}
\begin{remark}
In the original paper \cite{levin1998MLSapproximation}, the fill distance $h$ was defined slightly different. However, the two definitions are equivalent.
\end{remark}

\begin{Theorem}
Let $f$ be a function in $C^{m+1}(\Omega)$ with an $\hrho$ sample set. Then for fixed $\rho$ and $\delta$, there exists a fixed $k>0$, independent of $h$, such that the approximant given by equation \eqref{eq:basicMLS} is well conditioned (i.e., the least-squares matrix is invertible) for $\theta$ with a finite support of size $s =k h$. In addition, the approximant yields the following error bound:
\begin{equation*}
\norm{\tilde{p}(x) - f(x)}_{\Omega , \infty} <  M \cdot h^{m+1}
,\end{equation*}
where $\tilde{p}(x)$ is as defined in equation \eqref{eq:MLSfunctionApproximant}.
\label{thm:OrderMLSfunctions}\end{Theorem}

\begin{remark}
Although both Theorem \ref{thm:SmoothMLSfunctions} and Theorem \ref{thm:OrderMLSfunctions} are stated with respect to an interpolatory approximation (i.e., the weight function satisfies $\theta(0) = \infty$), the proofs articulated in \cite{levin1998MLSapproximation} are still valid taking any compactly supported non-interpolatory weight function. These proofs are based upon a representation of the solution to the minimization problem of Equation \eqref{eq:MLSfunctionApproximant} through a multiplication of smooth matrices. These matrices remain smooth even when the interpolatory condition is not met.
\label{rem:NonInterpolatoryTheta}\end{remark}
\begin{remark}
Notice that the weight function $\theta$ in the definition of the MLS for function approximation is applied on the distances in the domain. In what follows, we will apply $\theta$ on the distances between points in $\RR^n$ as we aim at approximating manifolds rather than functions. In order for us to be able to utilize Theorems \ref{thm:SmoothMLSfunctions} and \ref{thm:OrderMLSfunctions}, the distance in the weight function of equation \eqref{eq:basicMLS} should be $\theta(\norm{(x , 0) - (x_i, f(x_i))})$ instead of $\theta(\norm{x - x_i})$ (see Fig. \ref{fig:ApproximationOrderComment}). Nevertheless, as stated above, the proofs of both theorems as presented in \cite{levin1998MLSapproximation} rely on the representation of the solution to the minimization problem as a multiplication of smooth matrices. These matrices will still remain smooth after replacing the weight, as the new weighting is still smooth. Moreover, as explained in \cite{levin1998MLSapproximation} the approximation order remains the same even if the weight function is not compactly supported in case the weight function decays fast enough (e.g., by taking $\theta(r) \defeq e^{-\frac{r^2}{h^2}}$).
\end{remark}
\begin{figure}[ht]
\begin{centering}
\includegraphics[width={0.4\linewidth}]{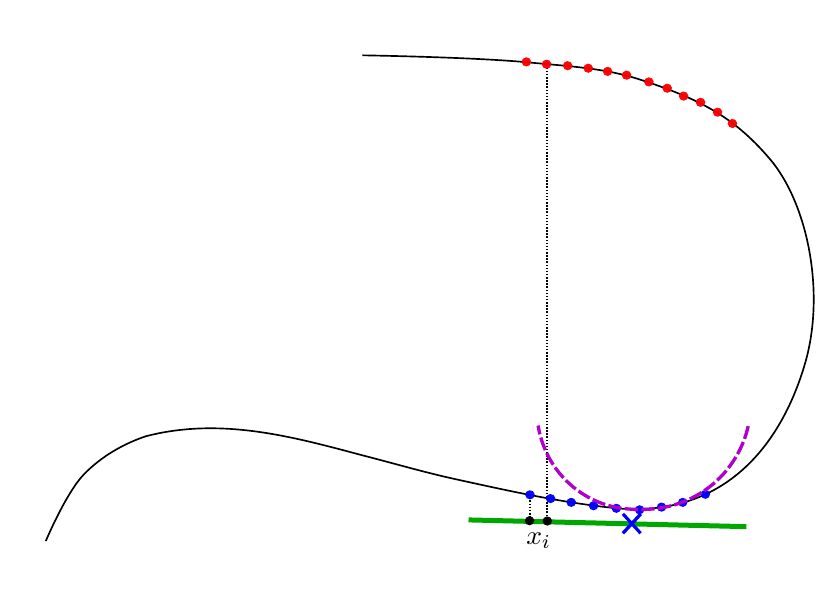}

\par\end{centering}

\caption{The effect of remote points when taking $\theta(\norm{(x , 0) - (x_i, f(x_i))})$ instead of $\theta(\norm{x - x_i})$. Assuming that the green line represents a given coordinate system around the point $ x $ (marked by the blue $ \times $), by taking the weights  $\theta(\norm{x - x_i})$ the contribution of both the red and blue samples to the weighted cost function would be $\mathcal{O}(h^{m+1})$. Alternatively, by taking $\theta(\norm{(x , 0) - (x_i, f(x_i))})$  with a fast decaying weight function the contribution of the red points would be negligible. Thus, the approximation (in purple) would fit the behavior of the blue points alone. \label{fig:ApproximationOrderComment}}

\end{figure}

\subsection{The MLS projection for surface approximation}

Following the rationale presented in \cite{levin2004mesh} let $S$ be an $n-1$ dimensional submanifold in $\RR^{n}$ (i.e., a surface), and let $\lbrace r_i \rbrace_{i=1}^{I}$ be points situated near $S$ (e.g., noisy samples of $S$). Instead of looking for a smoothing manifold, we wish to approximate the projection of points near $S$ onto a surface approximating $S$. This approximation is done without any prior knowledge or assumptions regarding $S$, and it is parametrization free.

Given a point $r$ to be projected on $S$ the projection comprises two steps: (a) finding a local approximating $n$-dimensional hyperplane to serve as the local coordinate system; (b) projection of $r$ using a local MLS approximation of $S$ over the new coordinate system. This procedure is possible since the surface can be viewed locally as a function.

\subsubsection*{The MLS projection procedure}
\textbf{Step 1 - The local approximating hyperplane}. Find a hyperplane \newline $H = \lbrace x | \langle a , x \rangle - D = 0 , x \in \mathbb{R}^n \rbrace \,,\, a \in \mathbb{R}^n \,,\, \| a \| = 1$, and a point $q$ on $H$ (i.e., $\langle a , q \rangle = D $), such that the following quantity is minimized over all $a \in \mathbb{R}^n, \| a \| = 1 , a = a(q) $ :

\begin{equation*}
I(q,a) = \sum_{i=1}^{I} (\langle a , r_i \rangle - D)^2 \theta(\| r_i - q\|) = 
\sum_{i=1}^I d(r_i , H)^2 \theta(\norm{r_i - q})
,\end{equation*}
where $\langle \cdot , \cdot \rangle$ is the standard inner product in $\mathbb{R}^n$, and $d(r_i , H)$ is the Euclidean distance between $r_i$ and the hyperplane $H$. Furthermore, $a(q)$ must be in the direction of the line that passes between $q$ and $r$, i.e.:
\begin{equation*}
(r-q) ~||~ a(q)
.\end{equation*}

\noindent\textbf{Step 2 - The MLS projection $P_m$} let $\lbrace x_i \rbrace_{i=1}^{I}$ be the orthogonal projections of the points $\lbrace r_i \rbrace_{i=1}^{I}$ onto the coordinate system defined by $H$, so that $r$ is projected to the origin. Referring to $H$ as a local coordinate system we denote the ``heights" of the points $\lbrace r_i \rbrace_{i=1}^{I}$ by \newline $f_i = \langle r_i , a \rangle - D$. We now wish to find a polynomial $p_0 \in \Pi_m^{n-1}$ minimizing the weighted least-squares error:
\begin{equation*}
p_0 = \argmin_{p \in \Pi_m^{n-1}} \sum_{i=1}^{I} (p(x_i) - f_i)^2 \theta(\| r_i - q\|)
.\end{equation*}
The projection of $r$ is then defined as
\begin{equation*}
P_m(r) \equiv q + p_0(0) a
.\end{equation*}
For an illustration of both Step 1 and Step 2 see Figure \ref{fig:MLSprojection}.

\begin{figure}[ht]
\begin{centering}
\includegraphics[width={0.4\linewidth}]{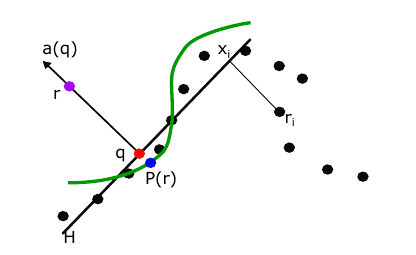}

\par\end{centering}

\caption{The MLS projection procedure. First, a local reference domain $H$ for the purple point $r$ is generated. The projection of $r$ onto $H$ defines its origin $q$ (the red point). Then, a local polynomial approximation $p_0(x)$ to the heights $f_i$ of points $r_i$ over H is computed. In both cases, the weight for each of the $r_i$ is a function of the distance to $q$ (the red point). The projection of $r$ onto $p_0$ (the blue point) is the result of the MLS projection procedure.\label{fig:MLSprojection}}

\end{figure}

As shown in \cite{levin2004mesh} the procedure described above is indeed a projection procedure (i.e., $P_m(P_m(r)) = P_m(r)$). Moreover, let $S\in C^{m+1}$ be the approximated surface and let $\tilde{S}$ be the approximating surface defined by the projection P, then it is expected that  $\tilde{S} \in C^{\infty}$ and the approximation order is $\OO(h^{m+1})$, where $h$ is the mesh size (tending to zero). The approximation order had been proven in \cite{alexa2003mesh.cont}, however, the $C^\infty$ result has not been proved prior to the current paper. In section \ref{sec:Theory} we present Theorems \ref{thm:ManifoldMMLS} and \ref{thm:OrderMMLS} which shows that the approximation is indeed a $C^\infty$ smooth manifold with approximation order of $\OO(h^{m+1})$ for a more general case. 

It is worth mentioning that the most challenging part of the algorithm is finding the approximating hyperplane (i.e., Step 1). The case is so since $a$ depends on $q$, and the weights are calculated according to the points' distance from $q$ which is a parameter to be optimized as well. It is, therefore, a non-linear problem. For the full implementation details see \cite{alexa2003mesh.cont}. An example of surface approximation performed with the MLS projection is presented in Figures \ref{fig:MLSsurface}.

\begin{figure}[ht]
\begin{centering}
\includegraphics[width={0.4\linewidth}]{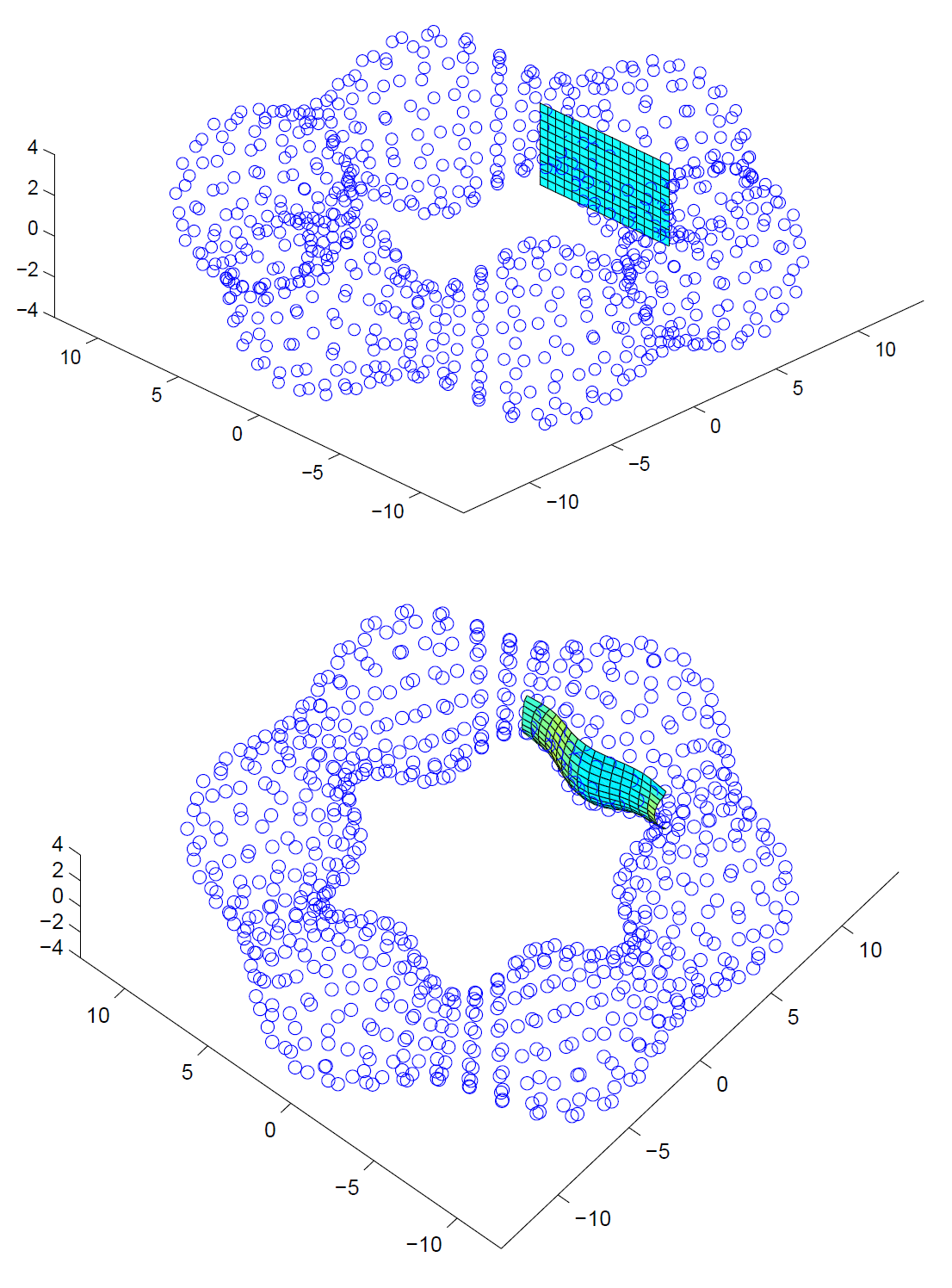}

\par\end{centering}

\caption{An example of the projection as appeared in \cite{levin2004mesh}: Upper part - data points and a plane segment $L$ near it. Lower part - the projection $P_2(L)$.\label{fig:MLSsurface}}

\end{figure}

\newpage
\section{MLS projection for manifolds (MMLS)}
\label{sec:ManifoldMLS}
The MLS procedure described in the previous section was designed for the case of unorganized scattered points in $\mathbb{R}^n$ lying near a manifold $\mathcal{M}$ of dimension $n-1$ (i.e., of co-dimension 1). Here we wish to extend the method to the more general case, where the intrinsic dimension of the manifold is $d$ (for some $d < n$). After presenting the generalized projection algorithm, we propose an implementation, whose complexity is linear in the ambient dimension $n$, and conclude with a theoretical discussion. 

\subsection{The MMLS projection}
Let $\MM$ be a manifold of dimension $d$ lying in $\RR^n$, and let the samples of $\MM$ hold the following conditions.
\subsubsection*{Clean Sampling Assumptions}
\label{sec:CleanSampling}
\begin{enumerate}
    \item $\MM\in C^2$ is a closed (i.e., compact and boundaryless) submanifold of $\RR^n$.
    \item $R = \{r_i\}_{i=1}^I\subset\MM$ is an $\hrho$ sample set with respect to the domain $\MM$ (see Definition \ref{def:h-rho-delta}).
\end{enumerate}
\subsubsection*{Noisy Sampling Assumptions}
\label{sec:NoisySampling}
\begin{enumerate}
    \item $\MM\in C^2$ is a closed (i.e., compact and boundaryless) submanifold of $\RR^n$.
    \item $\tilde R = \{\tilde r_i\}_{i=1}^I\subset\MM$ is an $\hrho$ sample set with respect to the domain $\MM$ (see Definition \ref{def:h-rho-delta}).
    \item $R = \{ r_i \}_{i=1}^I$ are noisy samples of $R$; i.e., $ r_i= \tilde r_i + n_i$.
    \item $\norm{n_i} < \sigma$
\end{enumerate}
Henceforth, whenever one of these two definition sets is met we shall state that the \textit{Clean Sampling Assumptions} or \textit{Noisy Sampling Assumptions} hold. 

Given a point $r$ near $\MM$ we define the Manifold Moving Least-Squares projection of $r$ through two sequential steps: 
\begin{itemize}
    \item[1.] Find a local $d$-dimensional affine space $H(r)$ around an origin $q(r)$ such that $H$ approximates the sampled points.
    Explicitly, $H = q + Span\{e_k\}_{k=1}^d$, where $\{e_k\}_{k=1}^d$ is some orthonormal basis of $\RR^d$. 
    $H$ will be used as a local coordinate system.
    \item[2.] Define the projection of $r$ using a local polynomial approximation $p:H \rightarrow \mathbb{R}^{n}$ of $\mathcal{M}$ over the new coordinate system. Explicitly, we denote by $x_i$ the projections of $r_i$ onto $H$ and then define the samples of a function $f$ by $f(x_i) = r_i$. Accordingly, the $d$-dimensional polynomial $p$ is an approximation of the vector valued function $f$. 
\end{itemize}

\begin{remark}
Since $\MM$ is a differentiable manifold it can be viewed locally as a function from the tangent space to $\RR^{n-d}$. It is therefore plausible to assume that we can find a coordinate system $H$ and refer to the manifold $\MM$ locally as a graph of some function $f:H\rightarrow \RR^{n-d}$ (see Lemma \ref{lem:HapproximationNoise} for a formal discussion regarding this matter).
\end{remark}
\begin{remark}
We would like the points $r$ to be projected onto a smooth $d$-dimensional manifold approximating $\MM$. In order to achieve this $H$ should depend smoothly on $r$ (see Theorem \ref{thm:SmoothCoordinates})
\end{remark}
\begin{remark}\label{rem:grassPush}
Throughout the paper, whenever we encounter an affine space
$$L = x + span\{e_k\}_{k=1}^d,$$ 
we will denote its Grassmannian counterpart (i.e., the linear space without the shift by $x$) by 
\[
\GG L = span\{e_k\}_{k=1}^d
.\]
\end{remark}
\noindent\textbf{Step 1 - The local Coordinates} \\
Find a $d$-dimensional affine space $H$, and a point $q$ on $H$, such that the following constrained problem is minimized:
\begin{equation}
   J(r; q, H) = \sum_{i=1}^{I} d(r_i , H)^2 \theta(\| r_i - q\|) 
\label{eq:Step1Minimization}
\end{equation}
under the constraints
\begin{enumerate}
\item $r-q \perp H$  \label{init_constraint:perp}
\item $q\in B_{\mu}(r)$ \label{init_constraint:search}
\item $\#\left(R\cap B_{\sigma+ h}(q)\right) \neq 0$ \label{init_constraint:proximity}
,\end{enumerate}
where $d(r_i , H)$ is the Euclidean distance between the point $r_i$ and the affine subspace $H$, $\mu$ is some fixed number (we will elaborate on it further below), $B_\eta(x)$ is an open ball of radius $\eta$ around $x$, $h$ is the fill distance from the $\hrho$ set in the sampling assumptions.

\begin{remark}
For a later use, we introduce the notation $q = q(r)$ and $H=H(r)$. 
\end{remark}

We wish to give some motivation to the definition of the minimization problem portrayed above.
Constraint \ref{init_constraint:search} limits the search space to a neighboring part of the manifold, whereas constraint \ref{init_constraint:proximity}, narrows it further to the vicinity of the samples, and, thus, voids the possibility of achieving solutions with zero value of $J$ (caused by the fact that there are no sample  points in the support of $\theta$) for an illustration see Figure \ref{fig:uniqueCircle}.
The necessity in constraint \ref{init_constraint:perp} is less obvious though. 
Minimizing $J(r; q, H)$ without this constraint will just yield a local PCA approximation around an unknown point $q$ (see the Appendix for a detailed explanation about local PCA). 
The added constraint links the approximation to the point $r$, which we aim to project onto $\MM$, as well as generalizes the idea of the Euclidean projection onto a manifold.
Explicitly, in the theoretical case, we know that if we have a point $r$ ``close enough" to a given manifold $\MM$ there exists a unique projection $P(r)$ of the point $r$ onto $\MM$.
In addition, we know that this projection maintains $r-P(r)\perp T_{P(r)}\MM$, which is echoed in constraint \ref{init_constraint:perp} described in the minimization problem of Equation \eqref{eq:Step1Minimization}.
This concept of a unique projection domain is better expressed by the definition of reach as introduced in \cite{federer1959curvature} .
\begin{definition}[Reach]
The reach of a subset $A$ of $\RR^n$, is the largest $\tau$ (possibly $\infty$) such that if $x\in\RR^n$ and the distance, $dist(A,x)$, from x to A is smaller than $\tau$, then $A$ contains
 a unique point, $P_A(x)\in A$, nearest to x. 
\end{definition}
From now on, whenever we refer to the \textbf{reach neighborhood of a manifold $\MM$} we mean:
\begin{equation}
     U_{reach} \defeq \{x\in\RR^n ~\vert~ dist(x, \MM) < rch(\MM)\}
\label{eq:ReachNeihborhood}\end{equation}
In our context, we refer to manifolds with positive reach, and we denote the reach of a manifold by $rch(\MM)$.
Accordingly, for a point $r$ in the reach neighborhood $U_{reach}$, there exists a unique projection $P_{\MM}(r)$ onto the manifold $\MM$.
As we show below in Lemma \ref{lem:HapproximationNoise}, the minimizers $q, H$ of Equation \eqref{eq:Step1Minimization} converge to $P(r), T_{P(r)}\MM$ respectively as the fill distance $h$ tends to zero (given some assumptions on the support of $\theta$) for $r$ in some neighborhood $U\subset U_{reach}$.

Therefore, we wish to generalize the concept of a reach neighborhood (relevant for the limit case) to a domain where our procedure yields a unique approximation.
In contrast to the $U_{reach}$ definition, we cannot take a neighborhood of the approximant prior to defining it.
Thus, constraint \ref{init_constraint:search} limits the search space around $r$, the point we wish to project.
This way, we avoid irrelevant and null solutions to the minimization problem.
We wish to stress that the noise level $\sigma$ in our sample set does not necessarily bound the environment within which we can solve the minimization problem.
For example in Lemma \ref{lem:HapproximationNoise} the noise level decays to zero in the order $\OO(h)$, but the uniqueness of the MMLS projection procedure is guaranteed in a neighborhood of a fixed size; explicitly, for points $r$ such that $d(r,\MM)<rch(\MM)/4$.

\begin{assumption}[Uniqueness Domain]
We assume that there exists an $\epsilon$-neighborhood of the manifold
\begin{equation*}
    U_{unique} \defeq \{x\in\RR^n ~\vert~ dist(x, \MM) < \epsilon < rch(\MM) \}
,\end{equation*} 
such that for any $r \in U_{unique}$ the minimization problem \eqref{eq:Step1Minimization} has a unique local minimum $q(r) \in B_{\mu}(r)$, for some constant $\mu < rch(\MM)/2$.
\label{eq:uniqueAssume}
\end{assumption}

Note that in order to achieve a unique solution for a given $r$ and avoid null solutions (i.e., points $q$ which has no samples in the support of $\theta$ around them) the decay of $\theta$ should be bounded from below, and $\mu$ should be large enough such that $P_\MM(r)\in B_\mu(r)$.
Figure \ref{fig:uniqueCircle} illustrates the reach neighborhood of a section of a circle restricting $r$ such that $d(r,\MM)<rch(\MM)/4$ and setting $\mu=rch(\MM)/2$.
To some extent, the circle example ``bounds" the behavior of the data in every $2d$ section of the manifold, as the reach bounds the sectional curvature.
An illustration of a uniqueness domain for a cleanly sampled curve embedded in $\RR^3$ can be seen in Figure \ref{fig:UniqueDomain}.

\begin{figure}[ht]
\begin{centering}
\includegraphics[width={0.6\linewidth}]{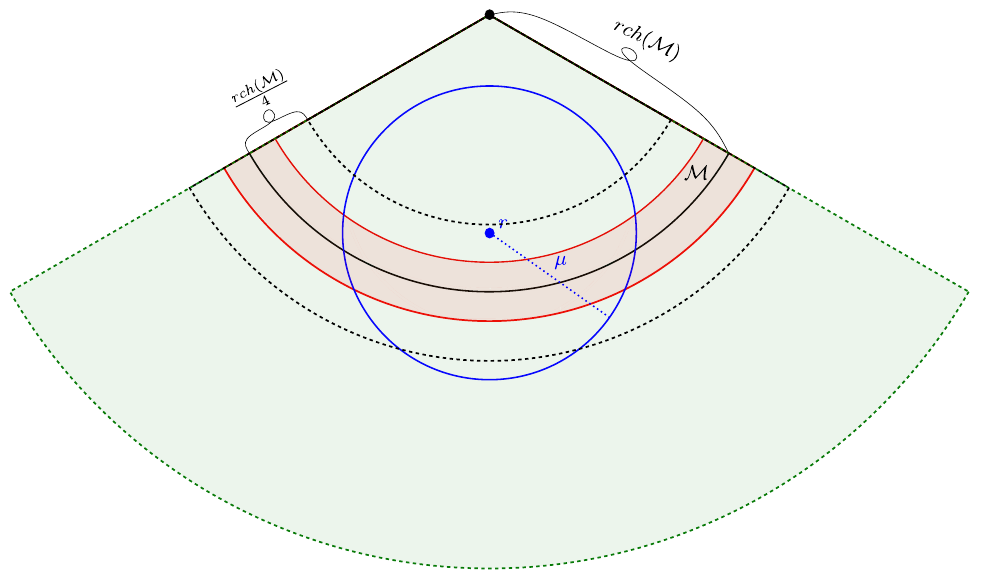}
\par\end{centering}
\caption{An illustration of a uniqueness domain on a circle section where we take $r$ such that $d(r,\MM)< rch(\MM)/4$ and set $\mu=rch(\MM)/2$. The black dot above is the center of the circle; the green region is the reach neighborhood of $\MM$; the red region is the noisy region from which we sample the manifold (i.e., the support of the distribution of sample points); the blue ball is the search region defined in constraint \ref{init_constraint:search} of Equation \eqref{eq:Step1Minimization}.  \label{fig:uniqueCircle}}
\end{figure}

\begin{figure}[ht]
\begin{centering}
\includegraphics[width={1\linewidth}]{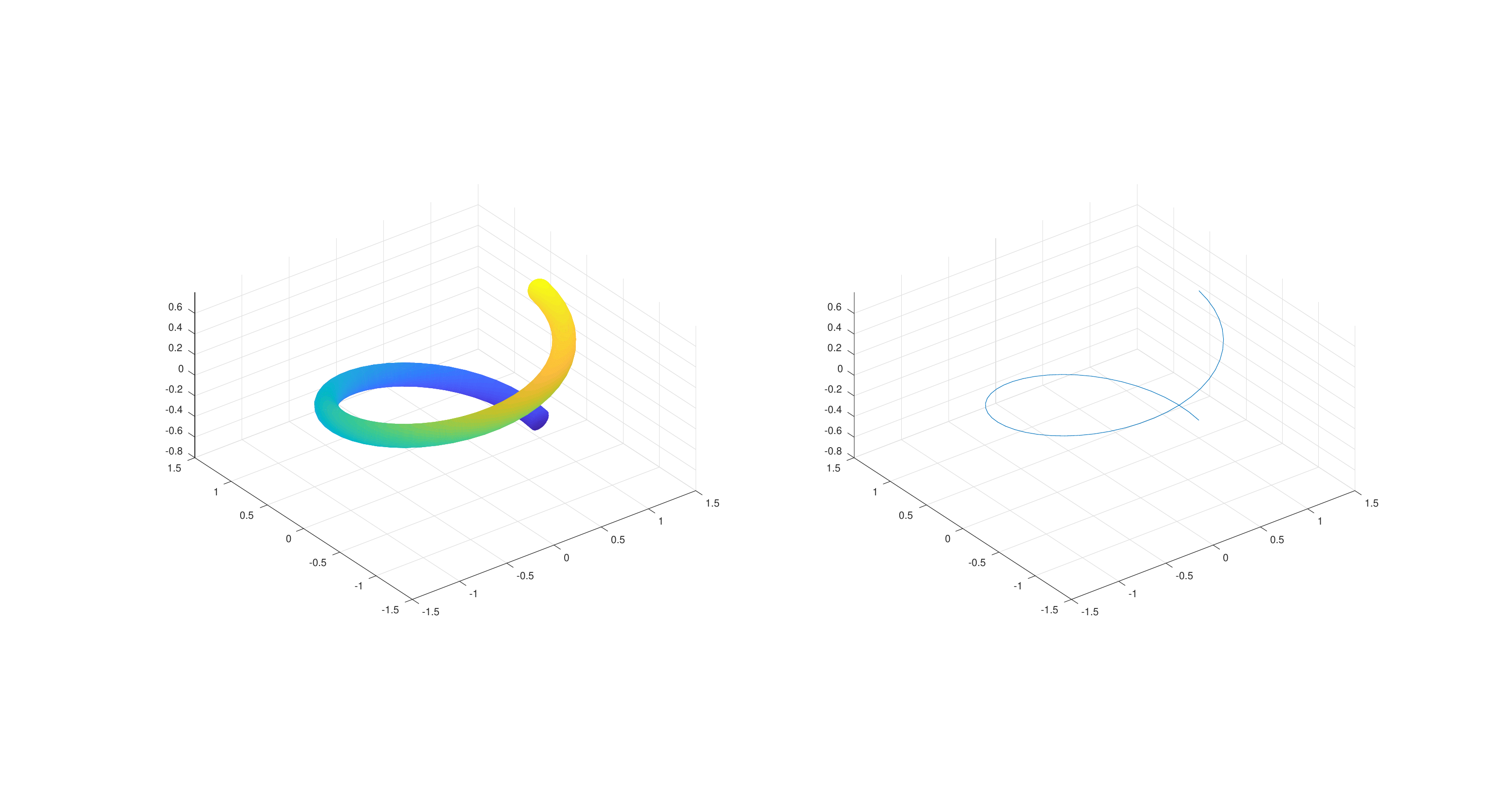}
\par\end{centering}
\caption{An illustration of a uniqueness domain. Right - a 1-dimensional manifold $\MM$ embedded in $\RR^3$. Left - a uniqueness domain $U$ of  $\MM$. \label{fig:UniqueDomain}}
\end{figure}

\vspace{5mm}\noindent\textbf{Step 2 - The MLS projection $P_m$.} Let $\lbrace e_k \rbrace_{k=1}^d$ be an orthonormal basis of $\GG H(r)$, and let $x_i$ be the orthogonal projections of $r_i$ onto $H(r)$ (i.e., $x_i = q(r) + \sum_{k=1}^d \langle r_i - q(r) , e_k \rangle e_k$). As before, we note that $r$ is orthogonally projected to the origin $q$. Now we would like to approximate $f:\RR^d \rightarrow \RR^n$, such that $f_i=f(x_i)=r_i$. The approximation of $f$ is performed by a weighted least-squares vector valued polynomial function $\vec{g}(x) = (g_1(x), ... , g_n(x))^T$ where $g_k(x) \in \Pi_m^d$ is a $d$-dimensional polynomial of total degree $m$ (for  $1 \leq k \leq n$).

\begin{equation}
\vec{g} = \argmin_{\vec{p} \in \Pi_m^{d}} \sum_{i=1}^{I} \| \vec{p}(x_i) - \vec{f}_i \|^2 \theta(\| r_i - q\|)
\label{eq:Step2}.\end{equation}
The projection $P_m(r)$ is then defined as:
\begin{equation}
P_m(r) = \vec{g}(0)
\label{eq:Step2Projection}\end{equation}

\begin{remark}
The weighted least-squares approximation is invariant to the choice of an orthonormal basis of $\RR^d$. 
\end{remark}

\begin{remark}
In fact we could have defined the second step as an approximation of a function $f:H\simeq \RR^d\rightarrow H^\perp\simeq\RR^{n-d}$.
Nevertheless, this would yield the exact same approximating object and the computational redundancy is negligible in the setting where $d\ll n$.
\end{remark}
\begin{remark}
In fact, considering each coordinate polynomial $g_k(x)$ separately we see that for all $1 \leq k \leq n$ we obtain the same system of least-squares just with different r.h.s. In other words, there is a need to invert (or factorize) the least-squares matrix only once! This fact is important for an efficient application of the implementation for high dimension $n$.
\end{remark}

\subsection{Implementation}
The implementation of Step 2 is straightforward, as this is a standard weighted least-squares problem. As opposed to that, minimizing \eqref{eq:Step1Minimization} is not a trivial task. Since the parameter $q$ appears inside the weight function $\theta$, the problem is non-linear with respect to $q$. We, therefore, propose an iterative procedure in which $q$ is updated at each iteration, and the other parameters are solved using a $d$-dimensional QR algorithm combined with a linear system solver.

\subsubsection*{Implementation of Step 1 - finding the local coordinates}
\label{alg:Step1}

We find the affine space $H$ by an iterative procedure. Assuming we have $q_j$ and $H_j$ at the $j^{th}$ iteration, we compute $H_{j+1}$ by performing a linear approximation over the coordinate system $H_j$. In view of the constraint $r-q\perp H$, we define $q_{j+1}$ as the orthogonal projection of $r$ onto $H_{j+1}$. We initiate the process by taking $q_{0} = r$ and solve a spatially weighted PCA around the point $r$ (for more details see \eqref{eq:geometricPCA} in the Appendix). This first approximation is denoted by $H_1$ and is given by the span of the first $d$ principal components $\lbrace u_k^1 \rbrace_{k=1}^d$. Thence, we compute:
\[q_1 = \sum_{k=1}^d \langle r - q_0 , u_k^1 \rangle u_k^1 + q_0 = q_0.\]
Upon obtaining $q_1 , H_1$ we continue with the iterative procedure as follows:
\begin{itemize}
\item Assuming we have $H_j , q_j$ and its respective frame $\lbrace u_k^j \rbrace_{k=1}^d$ w.r.t the origin $q_j$, we project our data points $r_i$ onto $H_j$ and denote the projections by $x_i$. Then, we find a linear approximation of the samples $f_i^j = f^j(x_i) = r_i$:
\begin{equation}\label{eq:alg_LS}
\vec{l}^j(x) = \argmin_{\vec{p} \in \Pi_1^{d}} \sum_{i=1}^{I} \| \vec{p}(x_i) - f_i^j \|^2 \theta(\| r_i - q_j\|)
.\end{equation}
Note, that this is a standard weighted linear least-squares as $q_j$ is fixed!
\item Given $\vec{l}^j(x)$ we obtain a temporary origin:
\[\tilde{q}_{j+1} = \vec{l}^j(0).\]
Then, around this temporary origin we build a basis $\hat{B} = \lbrace v_k^{j+1} \rbrace_{k=1}^d$ for $H_{j+1}$ with:
\[
v_k^{j+1} \defeq \vec{l}^j(u^j_k) - \tilde{q}_{j+1} ~~,~~ k=1,...,d
\]
We then use the basis $\hat{B}$ in order to create an orthonormal basis $B = \lbrace u_k^{j+1} \rbrace_{k=1}^d$ through a $d$-dimensional $QR$ decomposition, which costs $\OO(nd^2)$ flops. Finally we derive
\[q_{j+1} = \sum_{k=1}^d \langle r - \tilde{q}_{j+1} , u_k^{j+1} \rangle u_k^{j+1} + \tilde{q}_{j+1}.\]
This way we ensure that $r - q_{j+1} \perp H_{j+1}$.

\end{itemize}
See Figure \ref{fig:HelixLinear} for the approximated local coordinate systems $H$ obtained by Step 1 on noisy samples of a helix.

\begin{remark}
Note that a possible option for the least square minimization of Equation \eqref{eq:alg_LS} is the zero polynomial (i.e., $\vec{p} = \vec{0}$). Thus, if we reach the theoretical minimum of \eqref{eq:Step1Minimization} at some point the linear approximation step cannot yield a result better than the zero polynomial. So, the theoretical minimum is, in fact, a ``fixed point" of the procedure.
\end{remark}

\begin{algorithm}
\caption{Finding The Local Coordinate System $(H(r),q(r))$}
\label{alg:FindH}
\begin{algorithmic}[1]
\State {\bfseries Input:} $\lbrace r_i \rbrace_{i=1}^N, r, \epsilon$
\State{\bfseries Output:}\begin{tabular}[t]{ll}
                         $q$ - an $n$ dimensional vector \\
                         $U$ - an $n\times d$ matrix whose columns are $\lbrace u_j \rbrace_{j=1}^d$ 
                         \end{tabular}
                         \Comment{$H = q + Span\lbrace u_j \rbrace_{j=1}^d$}
\State define $R$ to be an $n\times N$ matrix whose columns are $r_i$
\State initialize $U$ with the first $d$ principal components of the spatially weighted PCA 
\State $q\leftarrow r$
\Repeat
    \State $q_{prev} = q$
    \State $\tilde{R} = R - repmat(q,1,N)$
    \State $\tilde{R} = \tilde{R} \cdot \Theta$ \Comment{where $\Theta = diag(\sqrt{\theta(\norm{r_1-q})}, \ldots, \sqrt{\theta(\norm{r_N-q})})$}
    \State $X_{N\times d} = \tilde{R}^T U$ \Comment{find the representation of $r_i$ in $Col(U)$}
    \State define $\tilde{X}_{N\times (d+1)} = \left[(1,...,1)^T, X\right]$
    \State solve $\tilde{X}^T\tilde{X}\alpha = \tilde{X}^T \tilde{R}^T$ for $\alpha \in M_{(d+1)\times n}$ \Comment{solving the LS minimization of $\tilde{X}\alpha \approx \tilde{R}^T$}
    \State $\tilde{q} = q + \alpha(1,:)^T$
    \State $Q, \hat{R} = qr(\alpha(2:end, :)^T)$ \Comment{where $qr$ denotes the QR decomposition}
    \State $U \leftarrow Q$
    \State $q = \tilde{q} + U U^T (r-\tilde{q})$
\Until {$\|q-q_{\text{prev}}\|<\epsilon$}
\end{algorithmic}
\end{algorithm}

\begin{algorithm}
\caption{Project $r$}
\label{alg:Projection}
\begin{algorithmic}
\State Input: $\lbrace r_i  \rbrace_{i=1}^N, r$
\State Output: $P_m(r)$
\State Build a coordinate system $H$ around $r$ using $\lbrace r_i \rbrace_{i=1}^N$ (e.g., via Algorithm \ref{alg:FindH})
\State Project each $r_i\in \RR^n$ onto $H \rightarrow x_i \in \RR^d$
\State Find the polynomial $p_r\in\Pi_m^d$ minimizing Equation \eqref{eq:Step2} using the samples $\lbrace (x_i, r_i)\rbrace_{i=1}^N$.
\State $P_m(r) \leftarrow p_r(0)$
\end{algorithmic}
\end{algorithm}

\subsubsection*{Complexity of the MMLS projection}

Since the implementation of Step 2 is straightforward, its complexity is easy to compute. The solution of the weighted least-squares for an $ m^{th} $ total degree $d$-dimensional scalar-valued polynomial, involves solving ${m+d \choose d}$ linear equations (since this is the dimension of $\Pi_m^d$), which is $\OO(d^m)$ equations for small $m$. Even though we are solving here for an $\RR^n$-valued polynomial the least-squares matrix is the same for all of the dimensions. Thus, the complexity of this step is merely $\OO(d^{3m} + n \cdot d^m)$. In addition, we need to compute the distances from the relative origin $q$ which costs $\OO(n \cdot I)$, where $I$ is the number of points. This can be reduced if we have a compactly supported weight function. Therefore, the overall complexity of the implementation of Step 2 is $\OO(n \cdot \tilde{I} + d^{3m} + n \cdot d^m)$, where $\tilde{I}$ is the number of points in the support of the weight function.

In a similar way, the complexity of each iteration of Step 1 involves $ \OO(n \cdot \tilde{I} + d^{3}) $ flops; from our experiments with the algorithm 2-3 iterations are sufficient to achieve good approximations (the entire numerical section was carried out using just 3 iterations). However, the initial guess of Step 1 involves a PCA which classically costs $\OO(n \cdot \tilde{I}^{2} )$. However, as the support should be determined such that the least-squares matrix is invertible we get that $\tilde{I} \propto \OO(d^m)$. Thus, we can use a randomized rank $d$ SVD implementation such as the one detailed in \cite{aizenbud2016SVD} and reduce the complexity of this step to $\OO(n \cdot \tilde{I} ) + \tilde{O}(n \cdot d^2)$, where $\tilde{O}$ neglects logarithmic factors of $d$. Plugging in the estimated size of $\tilde{I}$, we get that the overall complexity of Step 1 amounts to $\OO(n \cdot d^m)$

\begin{corollary}
The overall complexity for the projection of a given point $r$ onto the approximating manifold is $\OO( n \cdot d^m + d^{3m})$. Therefore, the approximation is linear in the ambient dimension $n$.
\end{corollary}
\begin{figure}[ht]
\begin{centering}
\includegraphics[width={0.6\linewidth}]{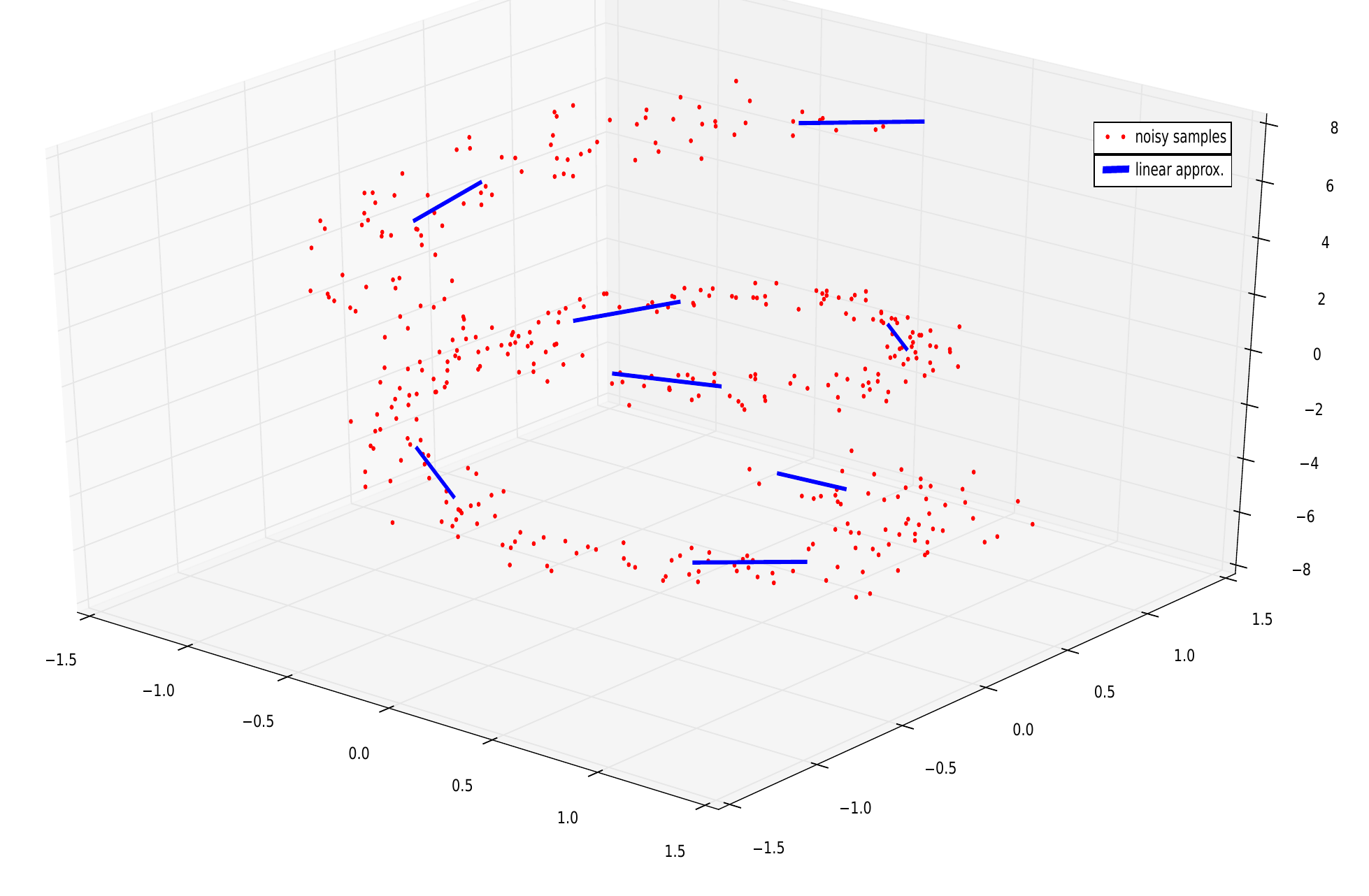}

\par\end{centering}

\caption{An approximation of the local coordinates $H(r)$ resulting from Step 1 implementation after three iterations, performed on several points $r$ near a noisy helix.\label{fig:HelixLinear}}

\end{figure}

\section{Theoretical analysis of the approximation}
\label{sec:Theory}
The main goal of the analysis presented in this section is the smoothness and approximation order theorems mentioned above (i.e., Theorems \ref{thm:ManifoldMMLS} and \ref{thm:OrderMMLS}). In the course of this analysis, we have built a theoretical connection between Least-Squares and PCA, discussed in Section \ref{sec:LSPCA}. This connection is being utilized in Section \ref{sec:HSmooth} as a tool for proving the smooth change of the coordinate system $H$. Nevertheless, the results reported in Section \ref{sec:LSPCA} are of general interest beyond the scope of this paper and will be discussed in a future publication. 

\subsection{Iterative least-squares and the approximation of the span of principal components}
\label{sec:LSPCA}
As a preparatory step for the proof of the smoothness of the affine sub-spaces $H(r)$ we first consider the following, simpler, iterative procedure.

Let $\lbrace r_i \rbrace_{i=1}^I$ be our sample set and let 
\[
\mathcal{R} = 
\left[
\begin{array}{ccc}
| &  & | \\
r_1 & \cdots & r_I \\
| &  & | 
\end{array}
\right]
.\]
Then, given an initial $d$-dimensional coordinate system
\[
U_0 = 
\left[
\begin{array}{ccc}
| &  & | \\
u^0_1 & \cdots & u^0_d \\
| &  & | 
\end{array}
\right]
,\]
we define the iterative least-squares procedure as:
\begin{enumerate}
\item Solve the linear least-squares problem
\[
A_{k+1} = \argmin_{A\in M_{n\times d}} \sum_{i=1}^I \norm{r_i - A x_i^k}^2 = \argmin_{A\in M_{n\times d}}\norm{\mathcal{R} - A X_k}_F^2
,\]
where 
\[
X_k = 
\left[
\begin{array}{ccc}
| &  & | \\
x^k_1 & \cdots & x^k_I \\
| &  & | 
\end{array}
\right]=
U_k^T \cdot \mathcal{R}
\] 
are the projections of $r_i$ onto $Col(U_k)$ the column space of $U_k$.

\item Apply Gram-Schmidt on the columns of $A_{k+1}$ to get a new orthogonal coordinate system. Namely,
\[
U_{k+1} \defeq Q(qr(A_{k+1}))
,\]
where $qr(A)$ is the QR decomposition of the matrix $A$ and $Q(qr(A))$ is the left matrix of this decomposition.
\end{enumerate}
In the following, we assume that the points are dense enough so that the least-squares is well posed at each iteration. Thus, using the aforementioned notation, the following proposition follows immediately
\begin{prop}
\[
A_{k+1} = \mathcal{R} \mathcal{R}^T U_k (U_k^T \mathcal{R} \mathcal{R}^T U_k)^{-1}
.\]
\end{prop} 
Hence, we get the following proposition as well
\begin{prop}
\[Col(A_{k+1}) = Col(\mathcal{R} \mathcal{R}^T U_k),\]
where $Col(A)$ is the column space of the matrix $A$.
\end{prop}

Furthermore, the columns of the matrix $U_{k}$, as defined in the second step of the iterations, are merely the result of applying the Gram-Schmidt process onto the matrix $ A_{k} $. Thus, the columns of $U_{k}$ are just some orthonormal basis of $Col(A_{k})$. 
As a result, since we are interested only in the column space of $A_k$, instead of solving the least-squares problem of the first step at each iteration, we can take a basis of $\mathcal{R} \mathcal{R}^T U_k$, which spans the exact same space. 
Explicitly, we can define the equivalent iterative procedure:
\begin{enumerate}
\item $U_{k+1} \defeq Q(qr(\mathcal{R} \mathcal{R}^T U_k))$
\end{enumerate}
Taking a close look at the newly defined iterations, it is apparent that it coincides with applying subspace iterations with respect to the matrix $\mathcal{R} \mathcal{R}^T$ \cite{stewart2001matrix}. Thus, the limit subspace achieved by this procedure would be the span of the first $d$ principal components of the matrix $\mathcal{R}$. Furthermore, if we denote the singular values of $\mathcal{R}$ by $\sigma_1 \geq ... \geq \sigma_d > \sigma_{d+1} \geq ... \geq \sigma_n$ we know that this process converges geometrically with a decay factor of magnitude $O\left(\abs{\frac{\sigma_{d+1}}{\sigma_{d}}}\right)$.
For a more elaborate proof and explanation of this discussion, we refer the readers to \cite{AizenbudLevinSober2019LS2PC,phdthesisSober}.

\subsection{Analysis of the MMLS projection}
We now define the approximating manifold as
\begin{equation}
    \tilde \MM \defeq \lbrace P_m(x) ~\vert~ x \in \MM \rbrace
,\end{equation}
where $P_m(x)$ is the MMLS projection described in equation \eqref{eq:Step2Projection}. In the next subsections we intend to show that this approximant, is a $C^{\infty}$ $d$-dimensional manifold, which approximates the original manifold up to the order of $\OO(h^{m+1})$, in case of clean samples. Furthermore, we show that $P_m(r) \in \tilde \MM$ for all $r$ close enough to the sampled manifold $\MM$. 

For convenience, we restate the problem presented in Step 1 and in Equation \eqref{eq:Step1Minimization}: given a point $r$ and scattered data $R = \lbrace r_i \rbrace_{i=1}^{I}$, find an affine subspace $H$ of dimension $d$ and an origin $q \in H$ which minimizes
\[
J(r; q , H) = \sum_{i=1}^{I} d(r_i , H)^2 \theta(\| r_i - q\|)
,\]
under the constraints
\begin{enumerate}
\item $r-q \perp H$ 
\item $q\in B_{\mu}(r)$ 
\item $\#\left(R\cap B_{\sigma+ h}(q)\right) \neq 0$
,\end{enumerate}
where $h$ is the fill distance of our sample set with respect to the domain $\MM$ (see the Sampling Assumption sets in Section \ref{sec:NoisySampling}).

\subsubsection{Some approximation results and motivation for Assumption \ref{eq:uniqueAssume}}
We start our inquiry by showing some initial approximation convergence properties for the minimization problem of Step 1 when the fill distance $h\rightarrow 0$, even without assuming the existence of a uniqueness domain.
An immediate result of the convergence would be that in the limit case (i.e. when $h\rightarrow 0$; or alternatively, when the sample set is the entire manifold) there exists a uniqueness domain.
Explicitly, we look at the given sample set as an instance from a family of sample sets refining with $h$.
We denote henceforth by $q^*_h(r)$ and $H^*_h(r)$ the solutions to the minimization problem of Equation \eqref{eq:Step1Minimization} with respect to a point $r$ and a sample set with a corresponding fill distance $h$.
In order to measure the difference between $H^*_h(r)$ and $T_{P(r)}\MM$, where $P(r)$ denotes the projection of $r$ onto $\MM$, we use the operator norm.
Explicitly, as we wish to know the difference in principal angles between these two affine spaces, we look at the difference in operator norm between the projections on their Grassmannian counterparts; i.e., if $H^*_h = q^*_h + span\{e_k\}_{k=1}^d$ and $T_{p}\MM = p + span\{e'_k\}_{k=1}^d$ then we measure the distance between the projections onto $\GG H^*_h = span\{e_k\}_{k=1}^d$ and $\GG T_p\MM = span\{e'_k\}_{k=1}^d$ by the operator norm:
\begin{equation}\label{eq:AngleBetweenAffine}
    \norm{P_{H^*_h(r)} - P_{T_p\MM}}_{op} \defeq 
    \norm{P_{\GG H^*_h(r)} - P_{\GG T_p\MM}}_{op} = 
    \max_{x\in \RR^n}\frac{\norm{(P_{\GG H^*_h(r)}- P_{\GG T_p\MM})x}}{\norm{x}}
,\end{equation}
and this is equivalent to measuring the maximal principal angle between the two affine spaces \cite{bjorck1973PrincipalAngles}.

Below, we show that as $h$ tends to zero  $q^*_h(r)\rightarrow P(r)$ and $H^*_h\rightarrow T_{P(r)}\MM$ (in the sense that $\norm{P_{H^*_h(r)} - P_{T_{P(r)}\MM} }_{op} \rightarrow 0$).
This convergence occurs in both the Clean Sampling Assumptions and Noisy Sampling Assumptions described in Section \ref{sec:CleanSampling} above.
However, the proofs deal only with the noisy case, as it encapsulates the results in the clean case as well.
In order to be able to show these properties, we add the demand that the noise bound $\sigma$ decays to zero as the fill distance $h\rightarrow 0$.

\begin{prop}
Let the Noisy Sampling Assumptions of Section \ref{sec:NoisySampling} hold, with the noise bounded by $\sigma = c_1 h$ and let $p\in\MM$. 
Then, for $c_2 \geq 2\sqrt{1+c_1^2} + (1+c_1)$ and small enough $h$ we have $\#(B_{c_2 h}(p)\cap R) \geq 2^d$. 
Furthermore, there exists a subset of $d$ points $r_{j} \in B_{c_2 h}(p)\cap R$ such that $ r_j - p$ are linearly independent.
\label{prop:PointsDistributionNoise}\end{prop}
\begin{proof}
Without limiting the generality, we set $p=0$. In case $\MM$ is flat, then it is a $d$-dimensional linear subspace of $\RR^n$. For convenience, let  $x\in\MM$ be written in some coordinate system as $x=(x_1, \ldots, x_d, 0, \ldots, 0)\in\RR^n$, and let $(x)_l$ denote the $l^{th}$ coordinate of a vector $x\in\RR^d$. 
We now look at the $h[2 \sqrt{1+c_1^2} + (1+c_1)] $ size neighborhood of $p$, or simply $B_{v}(0)$. 
Clearly, the grid points
\[p_{j}\in \left\lbrace x ~\left\vert~ (x)_l = \pm \frac{2 \sqrt{1+c_1^2}}{\sqrt{d}}h ~,~ \text{for } l=1,\ldots, d \right. \right\rbrace \] 
as well as the discs $B_{{(1+c_1)}h}(p_j)$ for $j=1,\ldots,2^d$ are contained in $B_{c_2h}(0)$ (see Figure \ref{fig:richSample} for an illustration).
Since $h$ is the fill distance, each disc $B_{(1+c_1)h}(p_j)$ contains at least one point $\tilde r$ of the set $\tilde R$ as well as its noisy version $ r =\tilde r + n$, as $\norm{n}<c_1 h$. 
Thus, $\#(B_{c_2h}(0)\cap R) \geq 2^d$ and there exists a set of $d$ linearly independent vectors in $B_{c_2h}(0)\cap R $, as required.
Now, going back to the case where the manifold is not flat, since $\MM\in C^2$ the distance between the tangent at $p$ to its $\OO(h)$ neighboring samples is $\OO(h^2)$, thus for a small enough $h$ the above argument holds.
\begin{figure}[ht]
\begin{centering}
\includegraphics[width={0.5\linewidth}]{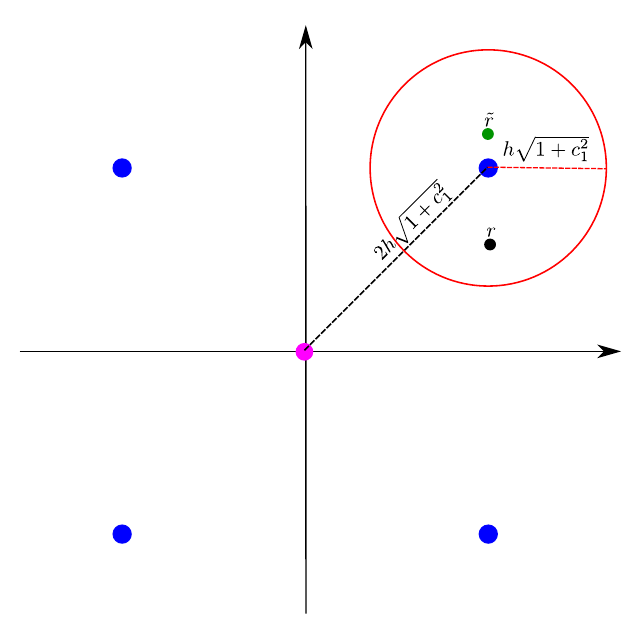}
\par\end{centering}

\caption{An illustration of the proof of Proposition \ref{prop:PointsDistributionNoise} in the flat domain.\label{fig:richSample}}

\end{figure}
\end{proof}

\begin{Lemma}[Convergence to the tangent]
Let the Noisy Sampling Assumptions of Section \ref{sec:NoisySampling} hold, and let the noise be bounded by $\sigma = c_1 h < \mu$, for some constant $c_1$ and $\mu$ of constraint \ref{init_constraint:search} of Equation \eqref{eq:Step1Minimization}. 
Let $U_{reach}$ be the reach neighborhood of $\MM$ \eqref{eq:ReachNeihborhood}, and the function $\theta(t)$ of Equation \eqref{eq:Step1Minimization} be monotonically decaying and compactly supported with $supp(\theta) = c_2 h$, where $c_2$ is some constant strictly greater than $2\sqrt{1 + c_1^2} + (1 + c_1)$. Suppose that $\theta(c_2h)>c_3>0$, for some constant $c_3$. 
Then, for $r$ such that $d(r, \MM) < rch(\MM)/4$ and $\mu= rch(\MM)/2$ (see Figure \ref{fig:uniqueCircle} for an illustration) we get
\begin{enumerate}
    \item as $h\rightarrow 0$ 
    \begin{equation}
 \norm{P_{H^*_h}(r_i)- r_i} = \Theta(h), ~~~ \forall r_i\in B_{c_2h}(q_h^*(r))\cap R
\label{eq:HapproximationNoise}\end{equation}

\item The following limits exist and
\begin{equation}
\lim_{h\rightarrow 0}q_h^*(r)= P(r); ~ \text{ and } ~\lim_{h\rightarrow 0}\norm {P_{H_h^*(r)} -P_{T_{P(r)}\MM}}_{op} =0
\label{eq:HapproximationNoise2}\end{equation}
\item If $q_h^*(r)= P(r) + \epsilon$ (where $\epsilon\in\RR^n$ and $\norm{\epsilon} = \eps$) then
\begin{equation}
\norm {P_{H_h^*(r)} -P_{T_{P(r)}\MM}}_{op} \leq \OO(h + \eps^2)
\label{eq:HapproximationNoise3}\end{equation}
\end{enumerate}

where $q^*_h, H^*_h$ denote a possible solution to the minimization of Equation \eqref{eq:Step1Minimization}.
(Note that we only assume the existence of a minimizer and do not demand its uniqueness as portrayed in Assumption \ref{eq:uniqueAssume})
\label{lem:HapproximationNoise}\end{Lemma}
\begin{proof}
We first notice that $q = P(r)$ coupled with $H = T_{P(r)}\MM$ maintain constraints \ref{init_constraint:perp}-\ref{init_constraint:proximity} of Equation \eqref{eq:Step1Minimization} since the projection onto $\MM$ keeps the condition 
\[
r - P(r) \perp T_{P(r)}\MM
,\] 
and constraint \ref{init_constraint:perp} is met. 
By the fact that $d(r, \MM) < rch(\MM)/4$ and the fact that $\mu = rch(\MM)/2$ we get that constraint \ref{init_constraint:search} is met.
In addition, since $h$ is the fill distance and $supp(\theta) = c_2h$, 
there exists some $ r_j\in R$ such that $\norm{\tilde r_j - P(r)} < h$, where $\tilde r_j$ is the clean version of $r_j$ as described in the Noisy Sampling Assumptions.
Since
\[
\norm{ r_j - P(r)} = \norm{ r_j - \tilde r_j + \tilde r_j- P(r)} \leq \norm{r_j -\tilde r_j} + \norm{\tilde r_j - P(r)} < c_1 h + h =  \sigma + h
,\]
we achieve 
\[
\#R\cap B_{ \sigma + h}(P(r)) \neq 0
,\] 
and constraint \ref{init_constraint:proximity} is met as well. 

Furthermore, since the tangent space is a first order approximation of a manifold $\MM\in C^2$, the cost function is compactly supported, and the sampling is a noisy version of an $\hrho$ set (see the definition of $\rho$  in Equation \eqref{def:rho}), then for all $x\in\MM$ (including $P(r)$) we have
\[
J(r; x, T_{x}\MM) =  \sum_{i=1}^I d^2( r_i, T_{x}\MM)\theta(\norm{ r_i - x}) \leq 
\]
\[\leq
  \sum_{i=1}^I d^2(\tilde r_i, T_{x}\MM)\theta(\norm{ r_i - x}) + 
  \sum_{i=1}^I \norm{ r_i - \tilde r_i}^2\theta(\norm{ r_i - x}) = 
  \OO(h^4) + \OO(h^2)
,\]
and so
\begin{equation*}
  J(r; x, T_{x}\MM) = \OO(h^2) \text{, as } h\rightarrow 0. 
\end{equation*}
Thus, as $h\rightarrow 0$ we get that the minimum $J(r; q_h^*(r), H_h^*(r)) = \OO(h^2)$ as well, and  
\begin{equation}
d(r_i, H^*_h) = \norm{P_{H^*_h}( r_i)-  r_i} = \OO(h) \label{eq:HapprpoxR}
\end{equation}
for $ r_i\in B_{c_2h}(q_h^*(r))\cap R$, since $\theta(c_2h)>c_3$.
Hence, we showed that \eqref{eq:HapproximationNoise} holds.

From constraint \ref{init_constraint:proximity} of Equation \eqref{eq:Step1Minimization} we know that for any given $h$ there must exist a point $ r_h\in  R$ such that $\norm{ r_h - q_h^*(r)} < \sigma + h = (1+c_1)h$. 
Furthermore, $ r_h =\tilde r_h + n_h$ for some $\tilde r_h\in \tilde R\subset\MM$ and $\norm{n_h} < c_1 h$ so for any given $h$ there exists $\tilde r_h\in \tilde R\subset \MM$ such that 
\[\norm{\tilde r_h - q_h^*(r)} =
\norm{\tilde r_h -r_h + r_h - q_h^*(r)} \leq 
\norm{\tilde r_h -r_h} + \norm{r_h - q_h^*(r)} < (2c_1 +1)h
.\]
Thus, since the manifold is closed, we get that there exists an accumulation point $p$ of $q_h^*(r)$ when $h\rightarrow 0$ (i.e., $r$ is fixed). 
Moreover, $p$ must be in $\MM$ as the distance $d(q^*_h(r), \MM) \leq \OO(h)$ tends to zero as $h\rightarrow 0$.
Let us look at a sequence $h_k\rightarrow 0 $ such that $q_{h_k}^*(r) = p + \epsilon_{h_k}$, and $\epsilon_{h_k}\rightarrow 0$, where $\norm{\epsilon_{h_k}} = \eps_{h_k}$.
Thus, for $ r_i\in B_{c_2h}(q_h^*(r))\cap R$ we get that 
\begin{equation*}
    \norm{\tilde r_i - P_{T_p\MM}(\tilde{r}_i)} = \OO(\eps_{h_k}^2)
,\end{equation*}
where $\tilde r_i$ are the clean versions of $r_i$.

Then, from the fact that locally the tangent is a linear approximation with a second order error term, and from Equation \eqref{eq:HapproximationNoise} we have that 
$$\norm{P_{H^*_{h_k}}(r_i) - P_{T_p\MM}(r_i)} = 
\norm{P_{H^*_{h_k}}(r_i) - r_i + r_i - P_{T_p\MM}(r_i)} \leq
\norm{P_{H^*_{h_k}}(r_i) - r_i} + \norm{r_i - P_{T_p\MM}(r_i)}$$
$$= \OO(h_k) + \norm{r_i - \tilde r_i + \tilde r_i - P_{T_p\MM}(r_i)} \leq 
\OO(h_k) + \norm{r_i - \tilde r_i} + \norm{\tilde r_i - P_{T_p\MM}(r_i)}$$
$$=\OO(h_k) + \OO(h_k) + \norm{\tilde r_i -  P_{T_p\MM}(\tilde r_i) + P_{T_p\MM}(\tilde r_i) - P_{T_p\MM}(r_i)}$$
$$\leq
\OO(h_k) + \norm{\tilde r_i -  P_{T_p\MM}(\tilde r_i)} + \norm{P_{T_p\MM}(\tilde r_i) - P_{T_p\MM}(r_i)}$$
$$\leq \OO(h_k) +\OO(\eps_{h_k}^2) + \norm{\tilde r_i - r_i} = \OO(h_k + \eps^2_{h_k}) $$
for points $r_i \in B_{c_2h_k}(q^*_{h_k}) \cap  R$.
Note that both projection operators $P_{H^*_{h_k}}$ and $P_{T_p\MM}$ are determined uniquely by $d$ linearly independent data points, as they are projections onto a $d$-dimensional affine spaces.
By Proposition \ref{prop:PointsDistributionNoise}, since $\theta(c_2h_k)>c_3$, we get that the data in $B_{c_2h_k}(p)\cap  R$ contain $d$ linearly independent points. 
Thus, the fact that $\norm{P_{H^*_{h_k}}(r_i) - P_{T_p\MM}(r_i)} = \OO(h_k + \eps^2_{h_k})$ indicates that
$ \norm {P_{H_{h_k}^*} -P_{T_{p}\MM}}_{op} = \OO(h_k + \eps^2_{h_k})$ by its definition in Equation \ref{eq:AngleBetweenAffine}.

Let us now show that all accumulation points $p$, which as explained above must be in $\MM$, must be exactly $P(r)$.
If $p\neq P(r)$ then we know that $r-p \not\perp T_p\MM$, as $r$ belongs to $U_{reach}$, for which there is a unique projection onto $\MM$. 
However, for all $h_k$ we have $r-q^*_{h_k} \perp H^*_{h_k}$, i.e., for all vectors $v_{h_k}\in \{x - q^*_{h_k}| x\in H^*_{h_k}\} $ we have $\langle r-q^*_{h_k}, v_{h_k} \rangle = 0$. 
Thus for all $v\in T_p\MM $ we can find $v_{h_k} \rightarrow v$ and 
$$\langle r-p, v \rangle = \lim_{h_k\rightarrow 0}\langle r-q^*_{h_k}, v_{h_k} \rangle = 0,$$
which contradicts the fact that $r-p \not\perp T_p\MM$.

As a consequence, we achieve that all accumulation points of $q^*_h$ must equal to $P(r)$ and so 
\begin{equation*}
q_h^*(r)= P(r) + \epsilon_h,~ (\eps_h\xrightarrow{h\rightarrow 0} 0); ~ \text{ and } ~\norm {P_{H_h^*(r)} -P_{T_{P(r)}\MM}}_{op} = \OO(h + \eps^2_{h})
,\end{equation*}
as required in Equations \eqref{eq:HapproximationNoise2} and \eqref{eq:HapproximationNoise3}.
\end{proof}

An immediate result of Lemma \ref{lem:HapproximationNoise} is that if our sample set is the entire manifold (i.e., when the clean samples are the entire manifold $\MM$), there exists a uniqueness domain $U_{unique}\defeq \{r | d(r, \MM) < rch(\MM)/4\}$ for the minimization problem of Equation \eqref{eq:Step1Minimization} (up to the fact that instead of sums we would have integrals). 
This result gives the motivation behind Assumption \ref{eq:uniqueAssume} for the discrete case, as a generalization of the reach neighborhood of the sampled data.

\subsubsection{Smoothness of the coordinate system $H$}
\label{sec:HSmooth}

 In this subsection, we aim at showing that the moving coordinate system produced by Step 1 of the MMLS algorithm is a smooth family with respect to the projected points $r$. 
 This part will enable our main results (Theorems \ref{thm:ManifoldMMLS} and \ref{thm:OrderMMLS}) regarding the MMLS projection in Section  \ref{sec:MMLSResults}.
 
 We set the focus at the beginning of this section to some properties of $H^*$, the approximating affine space.
 Then we show that the entire procedure is indeed a projection as expected. 
 And finally, we show that the approximating affine spaces $H^*(r)$ and their origins $q^*(r)$ are smooth families with respect to the projected point $r$. 

The following Lemma shows that if we fix the point $q$ we can define an affine space $H'(r;q)$ optimizing the function of Equation \eqref{eq:Step1Minimization}. Essentially, this affine space is achieved through Principal Component Analysis (PCA) of the data centered around $q$, after removing the direction of $r-q$.

\begin{Lemma}
Let the Noisy Sampling Assumptions of Section \ref{sec:NoisySampling} hold. Assume $q\in\RR^n$ is fixed, denote by $\{\tilde{w}_i\}_{i=1}^I$ the projections of $\{r_i - q\}_{i=1}^I$ onto the orthogonal complement of $Span\{r-q\}$, and let $\mathcal{R}$ be a matrix whose columns are $ \tilde{w}_i \cdot\sqrt{\theta(\norm{r_i - q})} $. Furthermore, assume that $rank(\mathcal{R}) > d$. Then, $H'$ minimizing the function $J(r;q,H)$ such that $q \in H'$ and $r-q \perp H'$, is determined uniquely by
\begin{equation*}
    H'(r; q) = q + Span \lbrace \vec{u}_k\rbrace_{k=1}^d
,\end{equation*}
where $\vec{u}_k$ are the leading principal components of the matrix $\mathcal{R}$.
In other words, the minimizing $H$ can be written as a function of $q$, i.e. as $H'(r; q)$.
\label{lem:Hofq}\end{Lemma}
\begin{proof}
Let $W$ be the affine $1$-dimensional subspace spanned by $r-q$. Specifically, we mean that $W = Span \lbrace r - q \rbrace + q$. Without loss of generality, we assume $q = \vec{0} \in \RR^n$ (otherwise we can always subtract $q$ and the proof remains the same) and therefore $H'$ is now a standard linear space around the origin. Accordingly, since $r - q = r$ the constraint mentioned above can now be rewritten as
\[ r \perp H .\]
So $W = Span \lbrace r \rbrace$ and we denote the projections of $\lbrace r_i \rbrace_{i=1}^{I}$ onto $W^{\perp}$ as $\lbrace w_i \rbrace_{i=1}^{I}$. Now let $\lbrace e_k \rbrace_{k=1}^n$ be an orthonormal basis of $\RR^n$ such that $\lbrace e_k \rbrace_{k=1}^d$ is a basis of $H$ and $e_{d+1} = \frac{r}{\norm{r}}$. Using this notation the minimization problem can be articulated as 
\[
J(r ; q , H) = \sum_{i=1}^{I} d(r_i , H)^2 \theta(\| r_i - q\|) = 
\sum_{i=1}^{I} \theta(||r_i||) \sum_{k=d+1}^n \abs{\langle r_i , e_k \rangle}^2
.\]
Looking closer at the inner product on the right hand side we get
\[
\sum_{k=d+1}^n \abs{\langle r_i , e_k \rangle}^2 = \norm{Q r_i - r_i}^2
,\]
where $Q$ is an orthogonal projection of $r_i$ onto $H$. Now since $H \subset W^\perp$ the first element of this summation 
\[
\abs{\langle r_i , e_{d+1} \rangle}^2 = \abs{\langle r_i , \frac{r}{\norm{r}} \rangle}^2
,\]
is invariant with respect to the choice of $H$. Thus we can reformulate the minimization problem as 
\[
\hat{J}(r;q,H) = \sum_{i=1}^{I} \theta(||r_i||) \sum_{k=d+2}^n \abs{\langle r_i , e_k \rangle}^2 = \sum\limits_{i=1}^{I} || P w_i - w_i ||^2 ~ \theta(||r_i||)
,\]
where $P$ is an orthogonal projection from $W^\perp$ onto $H$.
So in fact we wish to find a projection $P^\ast$ onto a $d$-dimensional linear subspace that minimizes the following:
\[
\sum\limits_{i=1}^{I} || P w_i - w_i ||^2 ~ \theta(||r_i||)
.\] 
From the discussion about the geometrically weighted PCA in the Appendix we know that the solution of the problem is given by taking the span of the first $d$ principal components of the matrix
\[
\mathcal{R} = 
\left[
\begin{array}{ccc}
| &  & | \\
w_1 \cdot \sqrt{\theta(\norm{r_1})} & \cdots & w_I \cdot \sqrt{\theta(\norm{r_I})} \\
| &  & | 
\end{array}
\right]
,\]
to be $H'$ - see equation \eqref{eq:geometricPCA}. In case $q\neq \vec{0}$ the matrix $R$ will be:
\begin{equation}\label{eq:RofQ}
\mathcal{R} = 
\left[
\begin{array}{ccc}
| &  & | \\
\tilde{w}_1 \cdot \sqrt{\theta(\norm{r_1 - q})} & \cdots & \tilde{w}_I \cdot \sqrt{\theta(\norm{r_I - q})} \\
| &  & | 
\end{array}
\right]
,\end{equation}
where $\lbrace \tilde{w}_i \rbrace_{i=1}^I$ are the projections of $\lbrace r_i - q \rbrace_{i=1}^I$ onto $W^\perp$. If we denote the singular value decomposition of $\mathcal{R}$ by $\mathcal{R} = U \Sigma V^T$ and $\vec{u}_i$ are the columns of the matrix $U$ (i.e., the eigenvectors of $\mathcal{R} \mathcal{R}^T$) then $H'$ is given explicitly by:
\[
H'(r; q) = Span \lbrace \vec{u}_i\rbrace_{i=1}^d
.\]
\end{proof}

\begin{remark}
Notice that the demand that the rank of the matrix $\mathcal{R}$ should be at least $d$ is, in fact, a demand on the local distribution of points to the tangent directions of $\MM$. Proposition \ref{prop:PointsDistributionNoise} shows that for small enough $h$ this demand is met for $q$ close enough to the manifold.
\end{remark}

This proof gives us intuition regarding the essence of the approximating affine subspace $H$. Explicitly, it is the span of the first $d$ principal components of the weighted PCA around the optimal $q$ with respect to the space $W^\perp$, where $W = q + Span \lbrace r- q \rbrace$. Thus, we can reformulate our minimization problem \eqref{eq:Step1Minimization}: for $r\in U$, find $q$ which minimizes 
\begin{equation}
J^\ast(r; q) = \sum_{i=1}^{I} d(r_i , H'(q))^2 \theta(\| r_i - q\|)
,\label{eq:JofQ}
\end{equation}
where $H'(q)$ is the affine space spanned by the first $d$ principal components of the geometrically weighted PCA around $q$ of the space $W^\perp$. Thus, the minimization problem is now with respect to $q$ alone. This simplifies the minimization task from the analytic perspective rather than the practical one, as the computation of PCA is costly when dealing with large dimensions. 

We now wish to tackle the question whether the approximant defined here is indeed a projection operator. In other words, can we say that we project an $n$ dimensional domain onto a $d$ dimensional one? In order for this to be true, we must demand that for a sufficiently small neighborhood, elements from $H^\perp$ are projected onto the same point (see Fig \ref{fig:HelixElipse} for an illustration). This result is articulated and proved in the following Lemma:

\begin{Lemma}
Let the Noisy Sampling Assumptions of Section \ref{sec:NoisySampling} hold. Let $r$ be in the uniqueness domain $U_{unique}$ of assumption \eqref{eq:uniqueAssume} and let $q^\ast(r)$ and $H^\ast(r)$ be the minimizers of $J(r;q,H)$ as defined above. Then for any point $\tilde{r} \in U_{unique}$ s.t. $\norm{\tilde{r} - q^\ast(r)}<\mu$ and $\tilde{r} - q^\ast(r) \perp H^\ast(r) $ we get $q^\ast(\tilde{r}) = q^\ast(r)$ and $H^\ast(\tilde{r}) \equiv H^\ast(r)$
\label{lem:ProjectionUniqueness}\end{Lemma}

\begin{proof}
Let us rewrite Equation \eqref{eq:Step1Minimization} in the form of Lagrange Multipliers to account for constraint 2 (the rest of the constraints deal with the neighborhood in which we search for the local minimum).
We first note that by taking some orthonormal basis $\{e_j\}_{j=1}^d$ on $GH$ (the Grassmannian counterpart of $H$) we can rephrase the term $d^2(r_i, H)$ from Equation \eqref{eq:Step1Minimization} as
\[
d^2(r_i, H) = \norm{r_i - P_H(r_i)}^2 = \norm{r_i - q}^2 - \norm{\sum_{j=1}^d\langle r_i - q, e_j\rangle e_j}^2 = \norm{r_i - q}^2 - \sum_{j=1}^d\abs{\langle r_i - q, e_j\rangle}^2
.\]

Thus, by setting $\{e_j\}_{j=1}^d$ to be directions in $\RR^n$ we can rephrase Equation \eqref{eq:Step1Minimization} along with constraint 2 as
\begin{equation}\label{eq:LagrangeFormJ}
\begin{array}{ll}
     J(r,q,e_1, ... , e_d, \Lambda) = &  \underbrace{\sum_{i=1}^I \left(\norm{r_i - q}^2 - \sum_{j=1}^d{\langle r_i - q, e_j\rangle}^2\right)\theta(\norm{r_i - q})}_{I} + \underbrace{\sum_{1\leq j < j' \leq d}\lambda_{j j'}\langle e_j, e_{j'}\rangle}_{II} + \\
     & \underbrace{\sum_{j=1}^d\lambda_{j j}(\langle e_j, e_{j}\rangle-1)}_{III} +
     \underbrace{\sum_{j=1}^d\lambda_{j}\langle r - q, e_{j}\rangle^2}_{IV}
\end{array}
,\end{equation}
where $\Lambda = (\lambda_1,\ldots,\lambda_d,\lambda_{11},\ldots,\lambda_{dd} )$.
Notice that terms $II$ and $III$ make sure that the minimizing directions should be some orthonormal basis; i.e., $\langle e_j, e_{j'} \rangle = \delta_{j j'}$, where $\delta_{j j'}$ is the Kronecker delta function.
And term $IV$ makes sure that constraint 2 is being maintained when the gradient of $J$ is null. 

Furthermore, the only terms in Equation \eqref{eq:LagrangeFormJ} which depend on $r$ is $IV$ and its partial derivatives $\frac{\partial J}{\partial q^k}, \frac{\partial J}{\partial e_j^k}$ with respect to the coordinates of $q$ and $e_j$ respectively.
Thus, if $(r; q^*(r), e^*_1(r), \ldots , e^*_d(r))$ is a critical point of $J$ and $\tilde r$ is another point such that $\tilde r - q^*(r) \perp H^*(r)$, then $(\tilde r; q^*(r), e^*_1(r), \ldots , e^*_d(r))$ is a critical point as well, as both $IV$ and its partial derivatives are still null.
Moreover, $J(r; q^*(r), e^*_1(r), \ldots , e^*_d(r)) = J(\tilde r; q^*(r), e^*_1(r), \ldots , e^*_d(r))$ as terms $II, III, IV$ of Equation \eqref{eq:LagrangeFormJ} are null and term $I$ is independent of $r$ and $\tilde r$.
Since $\tilde r\in U$, we get that $q^*(\tilde r) = q^*(r)$ and $H^*(\tilde r) = H^*(r)$ as well.

\end{proof}

\begin{figure}[ht]
\begin{centering}
\includegraphics[width={0.6\linewidth}]{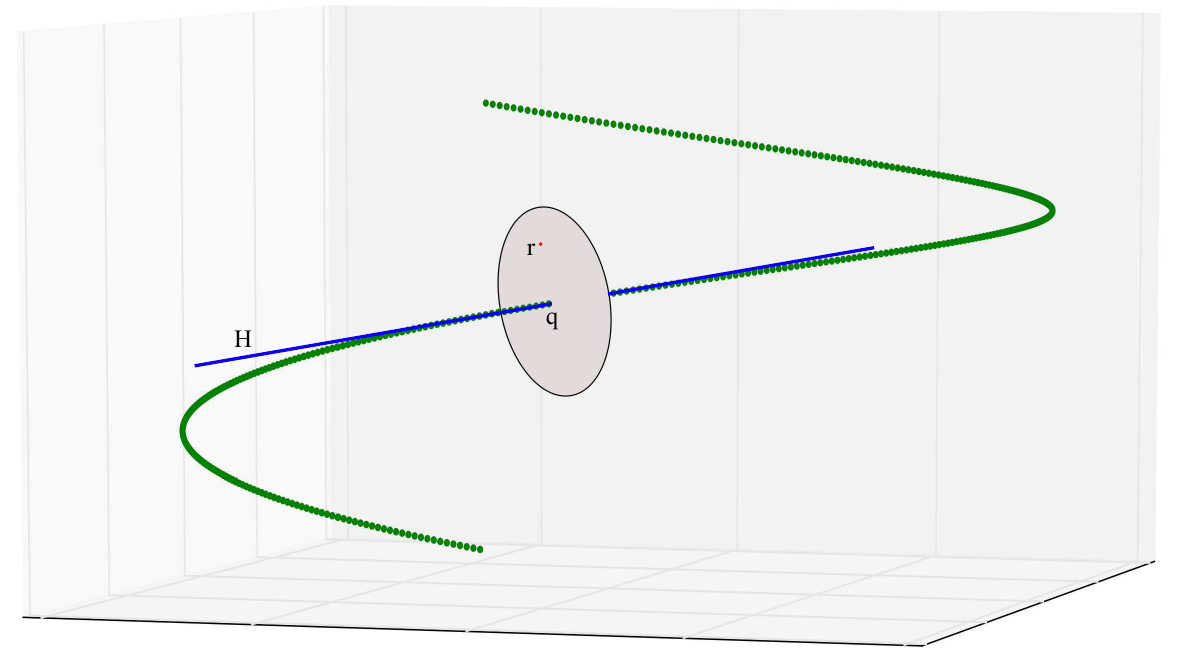}

\par\end{centering}

\caption{An illustration of a neighborhood of $q$ on $H^\perp$. All the points in this neighborhood should be projected to the same point.\label{fig:HelixElipse}}

\end{figure}

In order to be able to conduct an in-depth discussion regarding the smoothness of the approximant, and generalize the results quoted in Theorems \ref{thm:SmoothMLSfunctions} and \ref{thm:OrderMLSfunctions}, we introduce a definition of a smooth family of affine spaces.

\begin{definition}
\label{def:SmoothSpace}
Let $H(r)$ be a parametric family of $d$-dimensional affine sub-spaces of $\RR^n$ centered at the parameter $q(r)$. Explicitly, 
\[
w = q(r) + \sum_{k=1}^d c_k e_k(r) ~~,~~ \forall w\in H(r)
,\]
where $\lbrace e_k(r)\rbrace_{k=1}^d$ is a basis of the linear subspace $\GG H(r)$.
Then we say that the family $(H(r), q(r))$ changes smoothly with respect to $r$ if for any vector $v\in\RR^n$ the function 
\[
w(r) = q(r) + \sum_{k=1}^d \langle v - q(r), e_k(r) \rangle e_k(r)
,\]
describing the Euclidean projections of $v$ onto $H(r)$, vary smoothly with respect to $r$. 
\end{definition}

\begin{remark}
Note that this is equivalent to the demand that a family of projection operators $P_r$ varies smoothly in $r$ with respect to the operator norm (see Lemma \ref{lem:HapproximationNoise}, where the operator norm plays a key role).
\end{remark}

Ideally, we would have wanted to use the proof of Lemma \ref{lem:Hofq} and state that since the matrix $\mathcal{R}$ is smooth in $r$ and $q$ then there exists a smooth choice of $\mathcal{R} \mathcal{R}^T$'s eigenvectors (i.e., the basis $\{ u_k \}$ of $H'$). In the general case of smooth multivariate perturbations of matrices, this is not always true (e.g., see \cite{rainer2013perturbation, hinrichsen2005mathematical}). Nevertheless, in the following lemma, we are able to show that in our case the projections onto $H'(r; q)$ vary smoothly in both parameters. This is made possible due to the manifold structure as well as the utilization of the Iterative Least-Squares mechanism discussed in Section \ref{sec:LSPCA}. As shown there, this algorithm coincides with the famous subspace iterations, which is known to converge geometrically with a decay factor of $\abs{\frac{\sigma_{d+1}}{\sigma_d}}$, where $\sigma_k$ denotes the $k^{{th}}$ singular value of the matrix $\mathcal{R}(q)$ of Equation \eqref{eq:RofQ} (see \cite{stewart2001matrix}).
Since the points $\{r_i\}$ are samples of a $d$-dimensional manifold with additive noise, it is reasonable to assume that, on a local level, the variance of the data is significantly more dominant in its $d$ leading principal components than the others. Therefore, we add this demand to the following Lemma.

\begin{Lemma}
Let the Noisy Sampling Assumptions of Section \ref{sec:NoisySampling} hold.
Let $\theta(t) \in C^\infty$ be a compactly supported weight function with a support of size $\OO(h)$. Let the distribution of the data points $\lbrace r_i \rbrace_{i=1}^{I}$ be such that the minimization problem of $J(r;q,H)$ is well conditioned locally (i.e., the local least-squares matrices are invertible), and the the matrix $\mathcal{R}(q)$ of Equation \eqref{eq:RofQ} has singular values $\sigma_d > \sigma_{d+1}$ for all $(r,q(r))\in U \times \{x ~\vert \norm{x-q^*(r)} < \epsilon\}$, where $U$ is the uniqueness domain of Assumption \ref{eq:uniqueAssume}. In addition, let $H'(r; q)$ be the affine subspace minimizing $J(r;q,H)$ for a given $q$ as described in Lemma \ref{lem:Hofq}. Then $H'(r; q)$ varies smoothly ($C^{\infty}$) with respect to $q$ and $r$ for $(r,q(r))\in U \times \{x ~\vert \norm{x-q^*(r)} < \epsilon\}$
\label{lem:HsmoothinQ}\end{Lemma}
\begin{proof}
From Lemma \ref{lem:Hofq} we know that
\[
H'(r;q) = q +Span\lbrace u_i \rbrace_{i=1}^d
,\]
where $u_i$ are the leading principal components of 
\[
\mathcal{R}(q) = 
\left[
\begin{array}{ccc}
| &  & | \\
w_1 \cdot \sqrt{\theta(\norm{r_1-q})} & \cdots & w_I \cdot \sqrt{\theta(\norm{r_I-q})} \\
| &  & | 
\end{array}
\right]
,\]
and $w_i \defeq Proj_{W^\perp}(r_i - q) = r_i - \langle r_i - q , r -q \rangle \cdot (r - q)$, where $W = Span\{r-q\}$.

We begin with the exploration of smoothness with respect to $q$ in the vicinity of $ q_0, r_0$. Let $B_\delta(q_0) = \{q\in \RR^n \vert \norm{q - q_0} < \delta\}$ be a ball of radius $\delta$ around $q_0$, and let $\delta$ be small enough such that $B_\delta(q_0) \subset \{x ~\vert \norm{x-q^*(r)} < \epsilon\}$. 
For $q_0$ we have the directions of the leading principal components of $\mathcal{R}(q_0)$:
\[
H'(r_0;q_0) = q_0 + Span\lbrace u_i^0\rbrace_{i=1}^d
.\]
Similarly, the subspace $H'(r_0; q)$ is given by the span of the leading principal components of $\mathcal{R}(q)$. However, in order to show the smoothness of $H'(r_0;q)$ we consider another iterative procedure to achieve them.
We set for all $q\in B_\delta(q_0)$ the initial directions $\lbrace u_i^0\rbrace_{i=1}^d$ and update them iteratively using the iterative least-squares algorithm described in Section \ref{sec:LSPCA}, with constant weights. Namely, we initially set 
\[
\tilde{H}^0(r_0; q) = q + Span\lbrace u_i^0\rbrace_{i=1}^d
,\]
and then iterate through the minimization:
\[
A^k(q) = \argmin_{A\in M_{n\times d}} \sum_{i=1}^I \norm{r_i - A x_i }^2\theta(\norm{r_i - q})
,\]
where $x_i$ are the projections of $r_i$ onto $\tilde{H}^{k-1}(r_0; q)$. Note, that $\theta(\norm{r_i - q})$ are fixed for any given $q$. Following this minimization we define 
\[
\tilde{H}^k(r_0;q) \defeq q + Span\lbrace col(A^k(q)) \rbrace
,\]
where $\lbrace col(A^k(q)) \rbrace$ are the columns of the matrix $A^k(q)$. 

We now refer the reader to the proof of Theorem \ref{thm:SmoothMLSfunctions} given in \cite{levin1998MLSapproximation} where the MLS approximation of functions is presented as a multiplication of smoothly varying matrices (under the assumption that $\theta\in C^\infty$). In the case of function approximation, discussed in \cite{levin1998MLSapproximation}, we have the same coordinate system for each point $x$ in the domain. Our case differs in this respect as each iteration can be considered as an MLS approximation (with no constant term), just with a varying coordinate system. Nevertheless, as the solution is represented as a product of smoothly varying matrices, a smooth change in the coordinate system will result in a smooth approximation still. Thus, given that $A^{k-1}(q)$ changes smoothly with respect to $q$, we achieve that $A^{k}(q)$ varies smoothly as well. Since the initial step $\tilde{H}^0(r_0, q)$ is constant with respect to $q$, it follows that $\tilde{H}^k(r_0;q)$ varies smoothly with respect to $q$ for all $k$. 

From the discussion in Section \ref{sec:LSPCA} we know that the iterative procedure converges in a geometrical rate with a decay factor of $\abs{\frac{\sigma_{d+1}}{\sigma_{d}}}$, where $\sigma_k$ is the k$^{th}$ singular value of the matrix $\mathcal{R}(q)$.
From the assumption that for all $q\in\{q ~\vert~ \norm{q-q_0}<\epsilon\}$ the singular values of the matrix $\mathcal{R}(q)$ maintain $\sigma_d > \sigma_{d+1}$, and the fact that singular values vary continuously (see for example \cite{lax2013linear, stewart1990perturbation}), we can bound $\abs{\frac{\sigma_{d+1}}{\sigma_{d}}(q)} < M < 1$ for all $q\in B_\delta(q_0)$. 
Thus, as we have a uniform bound for all $q$, we achieve uniform convergence of the iterative procedure. 
As a result, the limit of the process $H'(r_0;q)$ is smooth as well. 

Let us now refer to the case where $q_0$ is fixed and $r$ belongs to a neighborhood of $r_0$. 
Then $r-q_0$ changes smoothly with $r$ and accordingly so does the space $W^\perp$. Thus, the matrix $\mathcal{R}$ changes smoothly with $r$ and we can apply the same iterative least-squares mechanism to deduce that $H'(r;q)$ is smooth with respect to $r$ as well.

\end{proof}

Note, that the proof of Lemma \ref{lem:HsmoothinQ} can be reproduced in a more general setting of symmetric positive definite matrices. As this notion is outside the scope of this paper we just quote the resulting theorem without any further discussion.
\begin{Theorem}
Let $A(x)\in M_{n\times n}$ be a symmetric positive definite matrix, which depend smoothly on the variable $x\in\RR^N$. Let $\lambda_1(x) \geq \lambda_2(x) \geq ... \geq \lambda_n(x)$ be the eigenvalues of $A(x)$, and let $\lambda_{d}(x) > \lambda_{d+1}(x)$. Then, the eigenspace corresponding to the $d$ most dominant eigenvalues will vary smoothly with respect to $x$. 
\end{Theorem}

Going back to our analysis, we are now ready to show that our moving coordinate system $H^\ast(r)$ is as well a smooth family of affine spaces.
\begin{Theorem}[Smoothness of $q^*(r), H^*(r)$]
Let the Noisy Sampling Assumptions of Section \ref{sec:NoisySampling} hold. Let $\theta(x)\in C^{\infty}$, $H$ be a $d$-dimensional affine space around an origin $q$ and let $U_{unique}$ be the uniqueness domain of Assumption \ref{eq:uniqueAssume}. Let $q^\ast(r), H^\ast(r)$ be the minimizers of the constrained minimization problem \eqref{eq:Step1Minimization}
Let $J^\ast(r; q)$ be the function described in Equation \eqref{eq:JofQ}.
Then for all $r'\in U$ such that $\left( \frac{\partial^2 J^{\ast}}{\partial q_i \partial q_j}\right)_{ij} \in M_{n \times n}$ is invertible at $(r',q^*(r'))$ we get:
\begin{enumerate}
\item $q^\ast(r)$ is a smooth ($C^{\infty}$) function in a neighborhood of $r'$.
\item The affine space $H^\ast(r)$ changes smoothly ($C^{\infty}$) in a neighborhood of $r'$.
\end{enumerate}
\label{thm:SmoothCoordinates}\end{Theorem}
\begin{proof}
First we express the minimization problem under constraint 2 using the Lagrange multipliers:
\[
J'(r; q, H) = \sum_{i=1}^{I} d(r_i , H)^2 \theta(\| r_i - q\|) + \sum_{k=1}^d \lambda_k \langle r - q , e_k(H)\rangle
,\]
where $\lbrace e_k(H) \rbrace$ is an orthonormal basis of $H$. Since we know from Lemma \ref{lem:Hofq} that $q^\ast$ is as well the minimizer of $J^\ast(r;q)$, which is a function of $r$ and $q$ alone, we can write down:
\[
J^\ast(r; q) = J(r; q, H'(r; q)) = \sum_{i=1}^{I} d(r_i , H'(r; q))^2 \theta(\| r_i - q\|) + \sum_{k=1}^d \lambda_k \langle r - q , e_k(H'(r; q))\rangle
,\]
where $H'(r; q)$ is the affine space defined in Lemma \ref{lem:Hofq}. Specifically, we know that $r-q \perp H'(r; q)$, therefore, $q^\ast$ is the minimizer of:
\begin{equation*}
J^\ast(r; q) = \sum_{i=1}^{I} d(r_i , H'(r; q))^2 \theta(\| r_i - q\|) 
\end{equation*}
In addition, by Lemma \ref{lem:HsmoothinQ} we know that $H'(r; q)$ changes smoothly ($C^{\infty}$) with respect to both $r$ and $q$ for $(r, q(r))\in U\times \{\norm{q - q^*(r)} <\epsilon\}$. Accordingly, we get that $d(r_i,H'(r; q))\in C^{\infty}$, and thus $J^\ast(r; q) \in C^{\infty}$ with respect to $r$ and $q$ in this domain.
Let us now denote:
\[
\nabla_q J^\ast (r , q) = \left(\pdev{J^\ast}{q_1}(r , q) , \pdev{J^\ast}{q_2}(r , q) , ... , \pdev{J^\ast}{q_n}(r , q) \right)^T
,\]
where $q_i$ is the $i^{th}$ coordinate of the vector $q \in \RR^n$. Stated explicitly $\nabla_q J^\ast$ is a $C^{\infty}$ function of $2n$ variables:
\[
\nabla_q J^\ast (r , q):\RR^{2n} \rightarrow \RR^n
,\]
for all $(r, q(r))\in U\times \{\norm{q - q^*(r)} <\epsilon\}$.
Since $q^\ast$ minimizes $J^\ast(r , q)$ for a given $r$ we get:
\[
\nabla_q J^\ast (r , q^\ast) = 0
.\]
Moreover, from the theorem's assumption we have that at $(r', q^*(r'))$ the matrix
\[
\left( \pdev{(\nabla_q J^\ast)_i}{q_j} \right)_{ij} = 
\left( \frac{\partial^2  J^\ast}{\partial q_i \partial q_j}\right)_{ij} \in M_{n \times n}
\]
is invertible, thus we can apply the Implicit Function Theorem and express the set of points $x$ in a neighborhood of $q^*(r')$ that maintain $\nabla_q J^\ast (r , x) = 0$ as a smooth function of $r$, i.e., there exists a neighborhood of $r'$ where $x = g(r)\in C^\infty$  \cite{lang2012fundamentals}.
Since we assume that for all $r$ in our domain there exists a unique minimizer we conclude that $g(r) = q^*(r)\in C^{\infty}$.
Moreover, using Lemma \ref{lem:HsmoothinQ} it follows that $H^\ast(r)=H'(q^\ast(r))\in C^\infty$ as well.

\end{proof}

\begin{remark}
Note that $(r, q^*(r))$ is a local minimum as $r$ belongs to the uniqueness domain. As a consequence, the function $J^*$ is locally convex at that point. Thus, the condition that the Hessian of $J^*$ at $(r, q^*(r))$ is invertible implies that at this minimum there is no direction for which the second derivative vanishes.
So, when the condition is not met, there should exist a sectional curve of $J^*$ that has vanishing first and second derivatives.
When the data is sampled at random, this seems to be unlikely.
In any case, this condition can be verified numerically and in all of our experiments, this condition is met.
\end{remark}

\subsubsection{Smoothness and approximation order of the MMLS}
\label{sec:MMLSResults}
After establishing the fact that the coordinate system varies smoothly, we turn to the final phase of this discussion, which is the smoothness and approximation order arguments regarding the approximant, resulting from the two-folded minimization problem presented in equations \eqref{eq:Step1Minimization}-\eqref{eq:Step2}. Initially, we wish to approve the fact that the local coordinate system (found by the solution to the minimization problem) is a valid domain for the polynomial approximation performed in Step 2. Ideally, we would have liked to obtain the tangent space of the original manifold as our local coordinate system, in order to ensure the validity of our coordinate system. In Lemma \ref{lem:HapproximationNoise} above we establish the fact that $H^*$ approximates the tangent space even without the uniqueness domain of Assumption \ref{eq:uniqueAssume}.
Thus, our choice of coordinate can be considered as a feasible choice for a local domain. 

Below, we utilize the results articulated in the preliminaries section (i.e., theorems \ref{thm:SmoothMLSfunctions} and \ref{thm:OrderMLSfunctions}) to show that we project the points onto a $C^\infty$ manifold, and that given clean samples of $\MM$,  these projections are $\OO(h^{m+1})$ away from the original manifold $\MM$.

Prior to asserting the theorems which deal with approximation order and smoothness, we wish to remind the reader that the approximating manifold is defined as
\begin{equation}
\tilde \MM = \lbrace P_{m}(x) ~\vert~ x \in \MM \rbrace
\label{def:S}
,\end{equation} 
where $P_{{m}}(x)$ is the ${m}^{th}$ degree moving least-squares projection described in equations \eqref{eq:Step1Minimization}-\eqref{eq:Step2}. The following discussion will result in proving that $\tilde \MM$ is indeed a $d$-dimensional manifold, which is $C^\infty$ smooth and approximates the sampled manifold $\MM$. We prove that $\tilde \MM$ is a $d$-dimensional manifold by showing that $P_{{m}}:\MM\rightarrow\tilde \MM$ is diffeomorphic \textit{almost everywhere} (i.e., Lebesgue measure zero set).

For convenience, in the following Lemmas we want to refer to the affine space $H^*$ of Equation \eqref{eq:Step1Minimization}, as an element in the Grassmannian $Gr_d(\RR^n)$.
Accordingly, if $H^* = q^* + span\{e_k\}_{k=1}^d$ we denote its counterpart in the Grassmannian by $\GG H^* \defeq span\{e_k\}_{k=1}^d$.
Explicitly, we look at the pairs $(q^*, \GG H^*)$ as belonging to the product space $\RR^n\times Gr_d(\RR^n)$.

In some of the following Lemmas and Theorems intend to use the same set of conditions and definitions (inherited from Lemma \ref{lem:HapproximationNoise}).
To avoid unnecessary repetitions we wish to state these upfront and whenever they are utilized we will state that the \textit{injectivity conditions} hold.
\subsubsection*{Injectivity Conditions}
\label{sec:Injectivity}
\begin{enumerate}
    \item $\norm{n_i} < \sigma = c_1 h < \mu $, for some constant $c_1$ and $\mu$ from constraint \ref{init_constraint:search} of Equation \eqref{eq:Step1Minimization}.
    \item The function $\theta(t)$ of Equation \eqref{eq:Step1Minimization} is monotonically decaying and compactly supported with $supp(\theta) = c_2 h$, where $c_2$ is some constant greater than $2 \sqrt{1 + c_1^2} + (1+c_1)$. 
    \item Suppose that $\theta(c_2h)>c_3>0$, for some constant $c_3$.
    \item Set $\mu = rch(\MM)/2 $ in constraint \ref{init_constraint:search} of Equation \eqref{eq:Step1Minimization}.
    \item Let $r$ be such that $d(r,\MM) < rch(\MM)/4$
\end{enumerate}
Note that in order to have sufficient conditions for the injectivity results we need to limit the amount of noise below $\sigma$ to a level of $\OO(h)$.
This does not mean that the injectivity will necessarily break in a more noisy setting.
As a matter of fact, our experiments show that even in much noisier cases the approximant is still a manifold.
Nevertheless, the proofs below rely on the assumption that the noise decay as the fill distance tends to zero.

\begin{Lemma}[Injectivity of $(q^*, H^*)$]
Let the Noisy Sampling Assumptions of Section \ref{sec:NoisySampling} as well as the Injectivity Conditions of Section \ref{sec:Injectivity} hold.
Denote by $(q^*(p), H^*(p))\in \RR^n\times Gr_d(\RR^n)$ the minimizers of Equation \eqref{eq:Step1Minimization} for $p\in\MM$, with respect to the sample set $\tilde R$. 
Then there exists $h_0$ such that for all $h\leq h_0$ the map
\[
(q^*, \GG H^*):\MM\rightarrow\RR^n \times {Gr}_d(\RR^n)
,\]
is injective.
\label{lem:qInjectivity}\end{Lemma}
\begin{proof}
Let $p_1, p_2 \in \MM$, we wish to show that if $(q^*(p_1), \GG H^*(p_1)) = (q^*(p_2), \GG H^*(p_2)) = (q^*, \GG H^*)$ it immediately follows that $p_1 = p_2$. 
From Lemma \ref{lem:HapproximationNoise} we know that 
\[
\lim_{h\rightarrow 0} q^*_h = p_1;
\text{ and }
\lim_{h\rightarrow 0} q^*_h = p_2
.\]
Thus, in the limit $p_1 = p_2$ must exist.
Let $h$ be fixed, then $q^*_h = p_1 + \epsilon_1$ and $q^*_h = p_2 + \epsilon_2$, and we denote $\eps_1 = \norm{\epsilon_1}$, $\eps_2 = \norm{\epsilon_2}$.
Without limiting the generality assume $\eps_1 \leq \eps_2 \eqdef \eps$, then
\begin{equation*}
    \norm{p_1 - p_2} = \OO(\eps)
.\end{equation*}
Furthermore, from \eqref{eq:HapproximationNoise3} we know that
\begin{equation*}
    \norm {P_{H_h^*(r)} -P_{T_{p_1}\MM}}_{op} = \OO(h + \eps^2)
\end{equation*}

So, if we denote an orthonormal basis of $\GG T_{p_1}\MM$ by $\{e_k\}_{k=1}^d$ then there exists a basis $\{e'_k\}_{k=1}^d$ of $\GG H^*$ 
\[
e'_k = e_k + \delta_k
,\]
where $\norm{\delta_k}= \OO(h + \eps^2)$.
Furthermore, using Taylor expansion we know that
\[
p_2 = p_1 + \sum_{k=1}^d x_k e_k + n  
,\]
where $n\in \GG T_{p_1}\MM^\perp$.
And, if we denote $\vec x = (x_1,\ldots, x_d)$ then
\[\norm{n}= \OO(\norm{\vec x}^2)\]

Since $p_2 - q^* \perp H^*$ we get for $j=1,\ldots,d$
\begin{equation*}
    \langle p_2 - q^*, e'_j \rangle = 0 
\end{equation*}
\begin{equation*}
\langle p_1 + \sum_{k=1}^d x_k e_k + n - q^*, e_j + \delta_j \rangle = 0    
\end{equation*}
\begin{equation*}
\langle p_1 - q^*, e_j  + \delta_j\rangle + \langle \sum_{k=1}^d x_k e_k , e_j + \delta_j \rangle + \langle n, e_j + \delta_j \rangle = 0
,\end{equation*}
\begin{equation*}
\langle p_1 - q^*, e'_j \rangle + \langle \sum_{k=1}^d x_k e_k , e_j + \delta_j \rangle + \langle n, e_j + \delta_j \rangle = 0
.\end{equation*}
Since $p_1 - q^*\perp H^*$ as well we get that
\begin{equation*}
\langle \sum_{k=1}^d x_k e_k , e_j + \delta_j \rangle + \langle n, e_j + \delta_j \rangle = 0
.\end{equation*}
\begin{equation*}
x_j +\langle \OO(\norm{\vec x}) ,  \delta_j \rangle + \langle \OO(\norm{\vec x}^2), \delta_j \rangle = 0
,\end{equation*}
\begin{equation*}
    x_j + \OO(\norm{\vec x})\OO(h + \eps^2) +  \OO(\norm{\vec x}^2)\cdot\OO(h + \eps^2) = 0
\end{equation*}
and so for all $j=1,\ldots,d$
\begin{equation}\label{eq:xjFirst}
x_j =  \OO(h + \eps^2) [\OO(\norm{\vec x})+  \OO(\norm{\vec x}^2)]
.\end{equation}
On the other hand,
\begin{equation*}\label{eq:xjSecond}
    x_j = \langle p_2 - p_1, e_j\rangle =\langle\OO(\eps), e_j\rangle =\OO(\eps)
.\end{equation*}
Thus, we know that 
\[
\norm{x} = \OO(\eps)
.\]
From \eqref{eq:xjFirst} we get that for small enough $h, \eps, \norm{x} $ there exists constants $c_1, c_2, c_3 > 0$ such that
\[
\norm{x} \leq \sqrt{d} c_1 (h + \eps^2)(c_2\norm{x} + c_3\norm{x}^2)
,\]
which results in $\norm{x} = 0$ or
\[
\frac{1 - \sqrt{d}c_1(h+\eps^2)c_2}{\sqrt{d}c_1(h+\eps^2)c_3}\leq \norm{x}
.\]
However, in the latter,  the left hand side will tend to $\infty$ as $h$ and $\eps$ approach zero, whereas the right hand side is $\OO(\eps)$.
Therefore, $\norm{x}=0$ has to hold and so $$p_1 = p_2.$$
\end{proof}

We now wish to show that the entire MMLS projection is injective as well. 
To achieve this we wish to show that for small enough $h$ the MMLS projection keeps the points inside the uniqueness domain.
As can be seen in Equation \eqref{eq:Step2Projection} the second step of the MMLS procedure computes a least-squares polynomial and then takes its value at zero to be the projection of the point.
In other words, if we represent the approximating polynomial in the monomial basis of $\Pi_m^d$ - i.e., $\mathcal{B} = \{1, x_1, ..., x_d, x_1^2, x_1 x_2, ...\}$
\[
p(x) = \sum_{k}a_k \phi_k(x) ~,~ \text{for } \phi_k\in\mathcal{B}
,\]
then the projection is merely the constant term $a_0$.
Thus, we first show that as the fill distance $h$ tends to zero, the constant term in the local polynomial approximation tends to the average of the approximated points.
As a result, we get that for small enough $h$ the MMLS projection will yield a point in the uniqueness domain.
This fact will be the key to showing the injectivity of the procedure.

\begin{Lemma}[Least-Squares constant term convergence to average]
Let $\{f_i\}_{i=1}^I$ be a set of samples from a function $f:\RR^d\rightarrow \RR$ taken at locations $\{x_i\}_{i=1}^I\subset \RR^d$. 
Let $\theta_i = \theta(\norm{x_i})$ be a set of weights with a compact support of size $c h$. 
Let  $p^*\in \Pi_m^d$ be the minimizer of
\begin{equation}\label{eq:WeightedLS}
    p^* = \argmin_{p\in\Pi_m^d} \sum_{i=1}^I \norm{f_i - p(x_i)}^2 \theta_i
,\end{equation}
and let $a^*_0$ be the constant term of $p^*$ when it is written in the monomial basis $\mathcal{B}$.
Assume that the least-squares problem of Equation \eqref{eq:WeightedLS} is well conditioned (i.e., the least-squares matrix is invertible).
Then, as $h\rightarrow 0$ we have 
\begin{equation}\label{eq:ConstantTermAvg}
    a^*_0 = \sum_{i=1}^I f_i \omega_i + \OO(h)
,\end{equation}
where
\[
\omega_i = \frac{\theta_i}{ \sum_{i'=1}^I  \theta_{i'} }
\]
\end{Lemma}

\begin{proof}
Let us write the function that we wish to minimize:
\[
J = \sum_{i=1}^I \norm{f_i - p(x_i)}^2 \theta_i = 
\sum_{i=1}^I (f_i - a_0 - a_1\phi_1 (x_i) - ... - a_N\phi_N (x_i))^2 \theta_i
,\]
where $\phi_k$ are monomials in the basis $\mathcal{B}$.

The polynomial $p^*(x) = \sum_{k=0}^N a^*_k \phi_k(x)$ minimizing $J$ maintains $\nabla J(a^*_0,...,a^*_N) = 0$ and so
\[
0 = \frac{\partial J}{\partial a_0}(a^*_0,...,a^*_N) = -2 \sum_{i=1}^I (f_i - a^*_0 - a^*_1\phi_1 (x_i) - ...- a^*_N\phi_N (x_i))\theta_i
\]
\[
\sum_{i=1}^I (f_i - a^*_0 +\OO(h))\theta_i = 0
\]
\[
\sum_{i=1}^I f_i \theta_i + \OO(h)=  \sum_{i=1}^I  a^*_0 \theta_i
=  a^*_0 \sum_{i=1}^I  \theta_i 
,\]
and so
\[
a^*_0  = \sum_{i=1}^I f_i \cdot \omega_i + \OO(h)
\]
as required.
\end{proof}

\begin{corollary}
$a^*_0$ of Equation \eqref{eq:ConstantTermAvg} is nearly a weighted average of the samples for a small enough $h$.
\end{corollary}

Notice that $p_r:\RR^d\rightarrow \RR^n$ the polynomial minimizing Equation \eqref{eq:Step2} is comprised of a different scalar-valued polynomial at each coordinate $p_r^j$ for $j=1,...,n$.
Moreover, each polynomial $p_r^j:\RR^d \rightarrow \RR$ solves a minimization problem such as the one portrayed in Equation \eqref{eq:WeightedLS}.
Thus, each coordinate of the vector $p_r(0)$ is nearly a weighted average of the samples projected onto this coordinate, for a small enough $h$.

\begin{corollary}
For a small enough $h$ the vector $p_r(0)$ is nearly a weighted average the samples $r_i \in B_{c_1 h}(q^*)\cap \tilde R$ and, thus, it is $\OO(h)$ away from their convex hull. 
\label{cor:qinUnique}\end{corollary}

Since the samples are taken from a manifold (with an additive noise that decays as $h\rightarrow 0$), it is $\OO(h)$ away from some tangent space; i.e., a flat.
Therefore, the vector $p_r(0)$ will remain in $U_{unique}$ the uniqueness domain of Assumption \ref{eq:uniqueAssume} for a sufficiently small $h$, as $U_{unique}\subset U_{reach}$ and $\tilde R\subset U_{unique}$ for small enough $h$.

\begin{corollary}\label{cor:PmInUnique}
For $r$ in the uniqueness domain $U_{unique}$ and $h$ small enough we get that $P_m(r)\in U_{unique}$ as well. 
\end{corollary}

After achieving this we can now turn to prove that the MMLS projection $P_m$ is indeed injective.
\begin{Theorem}[MMLS injectivity]\label{thm:MMLSInjective}
Let the Clean Sampling Assumptions as well as the Injectivity conditions hold.
Then, there exists $h_0$, such that for all $h\leq h_0$
\begin{equation*}
    P_m:\MM\rightarrow \tilde \MM
,\end{equation*}
is injective.
\end{Theorem}
\begin{proof}
The mapping $P_m:\MM\rightarrow\tilde\MM$ can be broken into a composition of two mappings.
Namely,
\[
r \stackrel{\text{Step 1}}{\mapsto}  (q^*_r, \GG H^*_r) \stackrel{\text{Step 2}}{\mapsto} p_r(0) = \tilde r 
,\]
where $r\in\MM$, $(q^*_r, \GG H^*_r)\in\RR^n\times Gr_d(\RR^n)$ the minimizers of Equation \eqref{eq:Step1Minimization} , and  $p_r(0)\in\tilde\MM$ from Equation \eqref{eq:Step2Projection}.

From Lemma \ref{lem:qInjectivity} we get that the first step is injective, thus if we show that the second step is injective the proof is concluded.
Explicitly, we wish to show that if $(q_{r_1}^*, \GG H^*_{r_1})$ and $(q_{r_2}^*, \GG H^*_{r_2})$ map to the same $\tilde r$, then it follows that $(q_{r_1}^*, \GG H^*_{r_1}) = (q_{r_2}^*, \GG H^*_{r_2})$.
Assume that we have $(q_{r_1}^*, \GG H^*_{r_1})$ and $(q_{r_2}^*, \GG H^*_{r_2})$ mapping to the same $\tilde r$.
Thus, $P_{H^*_{r_1}}(\tilde r) = q^*_{r_1} $ and $P_{H^*_{r_2}}(\tilde r) = q^*_{r_2} $.
In other words 
$$\tilde r - q^*_{r_1} \perp H^*_{r_1},$$ 
and 
$$\tilde r - q^*_{r_2} \perp H^*_{r_2}.$$
However, by Corollary \ref{cor:PmInUnique} $\tilde r$ is in the uniqueness domain $U_{unique}$, and so, according to Lemma \ref{lem:ProjectionUniqueness} we get:
\[
(q_{r_1}^*, \GG H^*_{r_1}) = (q_{r_2}^*, \GG H^*_{r_2})
,\]
as required.
\end{proof}

\begin{remark}
The demand that $\theta(x)$ should be compactly supported can be relaxed to be a fast decaying weight function. However, working with this condition complicates the argumentation, hence, we preferred clarity over generality.
\label{rem:NonCompactTheta}
\end{remark}

Upon obtaining these results, we are now prepared to move on to one of the main results of this article. Namely, in the following theorem, we show that $\tilde \MM$ is an approximating $d$-dimensional manifold, which is $C^\infty$ smooth.

\begin{Theorem}[MMLS is a smooth manifold]
Let the Noisy Sampling Assumptions of Section \ref{sec:NoisySampling}, as well as the Injectivity Conditions of Section \ref{sec:Injectivity} hold. Let the data points be distributed such that the minimization problem of Equation \eqref{eq:Step1Minimization} is well conditioned locally (i.e., the least-squares matrix is invertible). 
Let $\left( \frac{\partial^2 J^{\ast}}{\partial q_i \partial q_j}\right)_{ij}$ be invertible at $(p,q^*(p))$ for all $p\in\MM$ (where $J^\ast(r; q)$ is the function described in equation \eqref{eq:JofQ}).
Then the MMLS procedure of degree ${m}$ described in equations \eqref{eq:Step1Minimization}-\eqref{eq:Step2} projects any $r\in U_{unique}$ onto $\tilde \MM$ an almost everywhere $d$-dimensional submanifold of $\RR^n$. 
Furthermore, $\tilde \MM$ is $C^{\infty}$ smooth.
\label{thm:ManifoldMMLS}\end{Theorem}

\begin{proof}
The proof comprises the following arguments:
\begin{enumerate}
\item $P_m:\MM \rightarrow \RR^n$ is a $C^\infty$ function.
\item $P_m$ is almost everywhere diffeomorphism; thus, $\tilde \MM$ is a smooth manifold almost everywhere.
\item $\forall r \in U_{unique},~ P_{m}(r) \in \tilde \MM $.
\end{enumerate}
Note that by Theorem \ref{thm:SmoothCoordinates} we know that $H^*(r)$ varies smoothly with respect to $r$.
And $H^*(r)$ is a varying coordinate system for the second minimization step \eqref{eq:Step2}, which is merely a weighted least-squares function approximation.
We now refer the reader to the proof of Theorem \ref{thm:SmoothMLSfunctions} given in \cite{levin1998MLSapproximation}, where the MLS approximation of functions is presented as a multiplication of smoothly varying matrices (under the assumption that $\theta\in C^\infty$).
In the case of function approximation, discussed in \cite{levin1998MLSapproximation}, we have the same coordinate system for each point $x$ in the domain.
Our case differs in the fact that for each point $r$ we have a different coordinate system.
Nevertheless, this coordinate system varies smoothly with respect to $r$, and thus, this multiplication of matrices from \cite{levin1998MLSapproximation} still varies smoothly.
As such, the procedure yields a smooth approximation and $P_m:\MM\rightarrow \RR^n$ is a smooth function, and the first claim is proven.

From Theorem \ref{thm:MMLSInjective}, we know that $P_m:\MM\rightarrow \tilde \MM$ is injective.
Since $P_m$ is smooth we can apply Sard's Theorem and get that the differential of $P_m$ is non-degenerate almost everywhere.
Thus, we can apply the Inverse Function Theorem and get that $P_m$ is diffeomorphic almost everywhere. Hence the second claim is achieved.

Finally, let $r\in U_{unique}$ and by Corollary \ref{cor:qinUnique} we know that $q^*(r)\in U_{unique}$ as well. 
Then, there exists a point $p\in\MM$ such that $p - q^*(r) \perp H^*(r)$, since $\norm{P_{H^*(r)} - P_{T_{P(r)}\MM}}_{op} = \OO(h)$ by Lemma \ref{lem:HapproximationNoise}.
As a result of Lemma \ref{lem:ProjectionUniqueness} we get that $P_m(r)\in\tilde \MM$. 

\end{proof}

Lastly, we show that given clean samples of $\MM$ we achieve that $\tilde \MM$ approximates $\MM$ up to the order of $\OO(h^{{m}+1})$. 

\begin{Theorem}[MMLS approximation order]
Let the Clean Sampling Assumptions of Section \ref{sec:CleanSampling} hold.
Assume further that $\MM\in C^{m+1}$.
Then, for fixed $\rho$ and $\delta$, there exists a fixed $k>0$, independent of $h$, such that the MMLS approximation for $\theta$ with a finite support of size $s =k h$ yields, for a sufficiently small $h$, the following error bound:
\begin{equation*}
\norm{\tilde \MM_m - \MM}_{\textrm{Hausdorff}} <  M \cdot h^{{m}+1}
,\end{equation*}
where 
\[
\norm{\tilde \MM_m - \MM}_{\textrm{Hausdorff}} = \max\lbrace\max_{s\in\tilde \MM_m} d(s, \MM), \max_{x\in\mathcal{M}}d(x, \tilde \MM_m)\rbrace
\]
$\tilde \MM_m$ is the $m^{th}$ degree MMLS approximation of $\MM$ and $d(p,\NN)$, is the Euclidean distance between a point $p$ and a manifold $\NN$.
\label{thm:OrderMMLS}\end{Theorem}
\begin{proof}
Let $r\in\MM$, then, from Equation \eqref{eq:HapprpoxR} in the proof of Lemma \ref{lem:HapproximationNoise}, $H^\ast(r)$  approximates the sample set $\lbrace r_i \rbrace$ up to the order of $\OO(h^2)$ in an $\OO(h)$ neighborhood of $q^\ast(r)$. Therefore, the projections of $r_i$ onto $H^\ast(r)$ are also an $\tilde{h}$-$\tilde{\rho}$-$\tilde{\delta}$, where $\tilde{h} = \OO(h)$ (and $\tilde{\rho}\approx \rho$, $\tilde{\delta}\approx \delta$) for $h$ small enough. 
According to Theorem \ref{thm:OrderMLSfunctions}, as the projection $P_{{m}}(r)\in \tilde \MM_m$ is merely a local polynomial approximation of $\MM$ we achieve that $P_{{m}}(r)$ is $\OO(h^{{m}+1})$ away from the manifold $\MM$. Accordingly, for all $r\in\MM,$ $d(r,\tilde \MM_m) \leq \OO(h^{{m}+1})$. Furthermore, for each $s\in \tilde \MM_m$ there exists a point $r\in\MM$ such that $s = P_{{m}}(r)$ which is $\OO(h^{{m}+1})$ away from $\MM$. Thus, for all $s\in\tilde \MM_m,$ $d(s,\mathcal{M}) \leq \OO(h^{{m}+1})$ as well, and the theorem follows. 
\end{proof}

\begin{remark}
Although entire Section \ref{sec:ManifoldMLS} was pronounced using the standard Euclidean norm, all of the definitions, development, and proofs are applicable for the general case of an inner product norm of the form $\norm{\cdot}_A = \sqrt{x^T A x}$. Where $A$ is a symmetric positive definite matrix. 
\end{remark}

\section{Numerical examples}
\label{sec:examples}
In this section we wish to present some numerical examples which demonstrate the validity of our method. In all of the following examples we have implemented Step 1 as described in Section \ref{alg:Step1} using just three iterations. The weight function utilized in all of the examples (and many others omitted for brevity) is $\theta(r) = e^{- \frac{r^2}{\sigma^2}}$, where the $\sigma$ was approximated automatically using a Monte-Carlo procedure:
\begin{enumerate}
\item Choose $100$ points from $\lbrace r_i \rbrace_{i=1}^I$ randomly
\item For each point:
    \begin{itemize}
    \item Calculate the minimal $\sigma$ such that the least-squares matrix is well conditioned (in fact we chose $10$ times more points than needed).
    \end{itemize}
\item Take the maximal $\sigma$ from the $100$ experiments.
\end{enumerate}

As stated above, although many of our claims rely upon a compactly supported weight function, the results still hold a weight function that decays fast enough (in our case we used an exponential decay). 

\subsection*{$1$-dimensional helix experiment}
In this experiment we have sampled 400 equally distributed points on the helix $(sin(t), cos(t), t)$ for $t \in [-\pi, \pi]$ (Fig. \ref{fig:HelixSteps}A) with uniformly distributed (between $-0.2$ and $0.2$) additive noise (Fig. \ref{fig:HelixSteps}B). In all of the calculations we have used the Mahalanobis norm, which is of the type $\sqrt{x^T A x}$, instead of the standard Euclidean. Assigning $d=1$ (i.e., the manifold's dimension), we projected each of the noisy points and the approximation can be seen in Fig. \ref{fig:HelixSteps}C. The comparison between the approximation and the original as presented in Fig. \ref{fig:HelixSteps}D speaks for itself. 
\begin{figure}[ht]
\begin{centering}
\includegraphics[width={0.9\linewidth}]{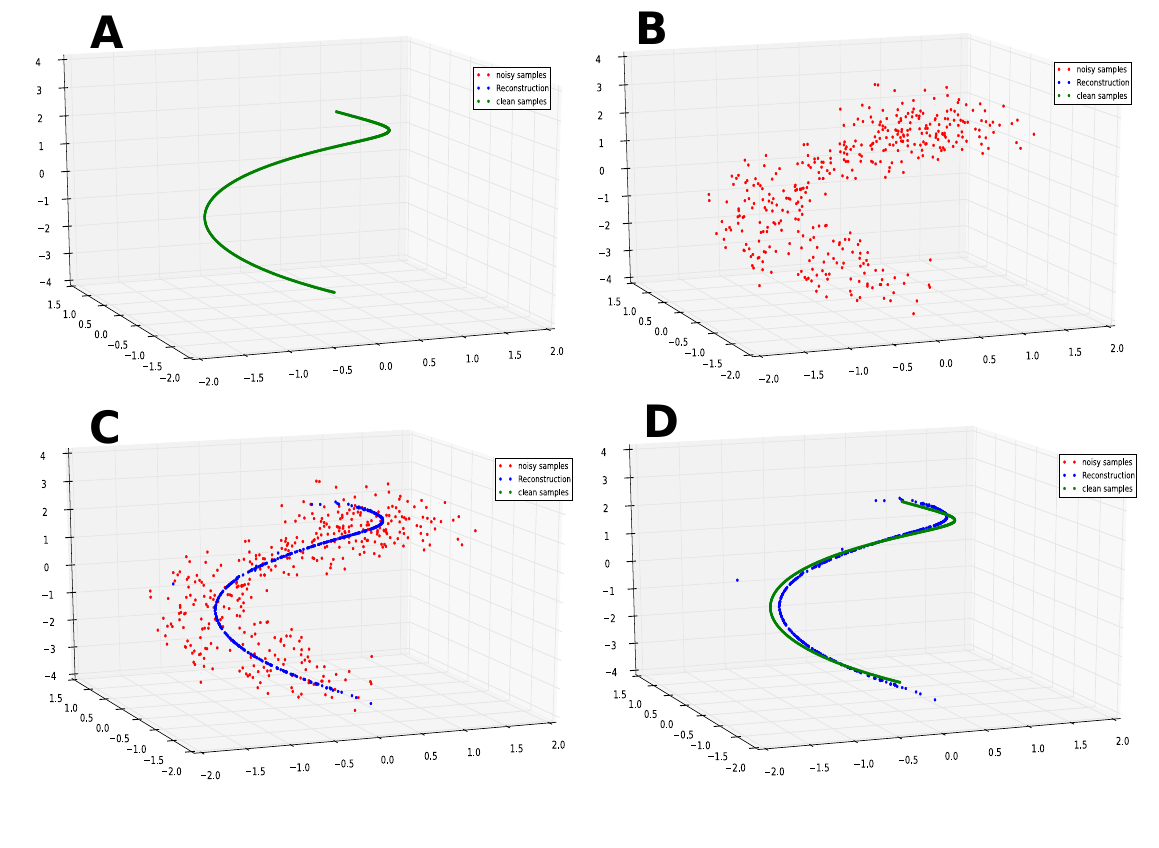}
\par\end{centering}
\caption{Approximation of $1$-dimensional helix. (A) clean samples (green); (B) noisy samples (red), after adding noise distributed $U(-0.2,0.2)$; (C) the approximation (blue) overlaying the noisy samples (red); (D) comparison between the approximation (blue) and the original clean samples (green)}
\label{fig:HelixSteps}\end{figure}

\subsection*{Ellipses experiment}
Here we sampled $144$ images of ellipses of size $100 \times 100$. The ellipses were centered and we did not use any rotations. Thus, we have $144$ samples of a $2$-dimensional submanifold embedded in $\RR^{10000}$. We have added Gaussian noise $\NN(0,0.05)$ to each pixel in the original images (e.g., see Fig. \ref{fig:NoisyEllipses}). One of the phenomena apparent in $n$-dimensional data is that if we have a very small random noise (i.e., bounded by $\epsilon$) entered at each dimension, the noise level (in the norm) is augmented approximately by a factor of $\sqrt{n}$. In our case, the noise bound is of size $100 \times 0.05 = 5$, whereas the typical distance between neighboring images is approximately $2.5-3$. Therefore, if we use the standard Euclidean norm the localization is hampered. In order to overcome this obstacle, we have used a $100$ dimensional distance. Explicitly, we have performed a pre-processing randomized SVD and reduced the dimensionality to $50$ times the intrinsic dimension. The reduced vectors were used just for the purpose of distance computation in the projection procedure process. Several examples of projections can be seen in Fig. \ref{fig:EllipsesProjections}. An example of the 2-dimensional mapping of the 144 samples projected onto $H$ is presented in \ref{fig:EllipsesH}.
\begin{figure}[ht]
\begin{centering}
\includegraphics[width={0.6\linewidth}]{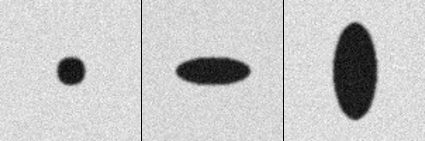}
\par\end{centering}
\caption{Examples of noisy ellipses. The noise is normally distributed $\NN(0,0.05)$ at each pixel.}
\label{fig:NoisyEllipses}\end{figure}

\begin{figure}[ht]
\begin{centering}
\includegraphics[width={0.6\linewidth}]{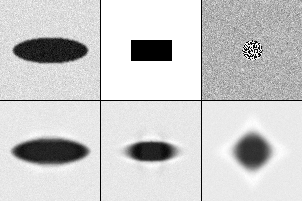}
\par\end{centering}
\caption{Projections on the ellipses $2$-dimensional manifold. Upper line: vectors that were projected. Lower line: the projections of the upper line on the ellipses manifold.}
\label{fig:EllipsesProjections}\end{figure}

\begin{figure}[ht]
\begin{centering}
\includegraphics[width={0.6\linewidth}]{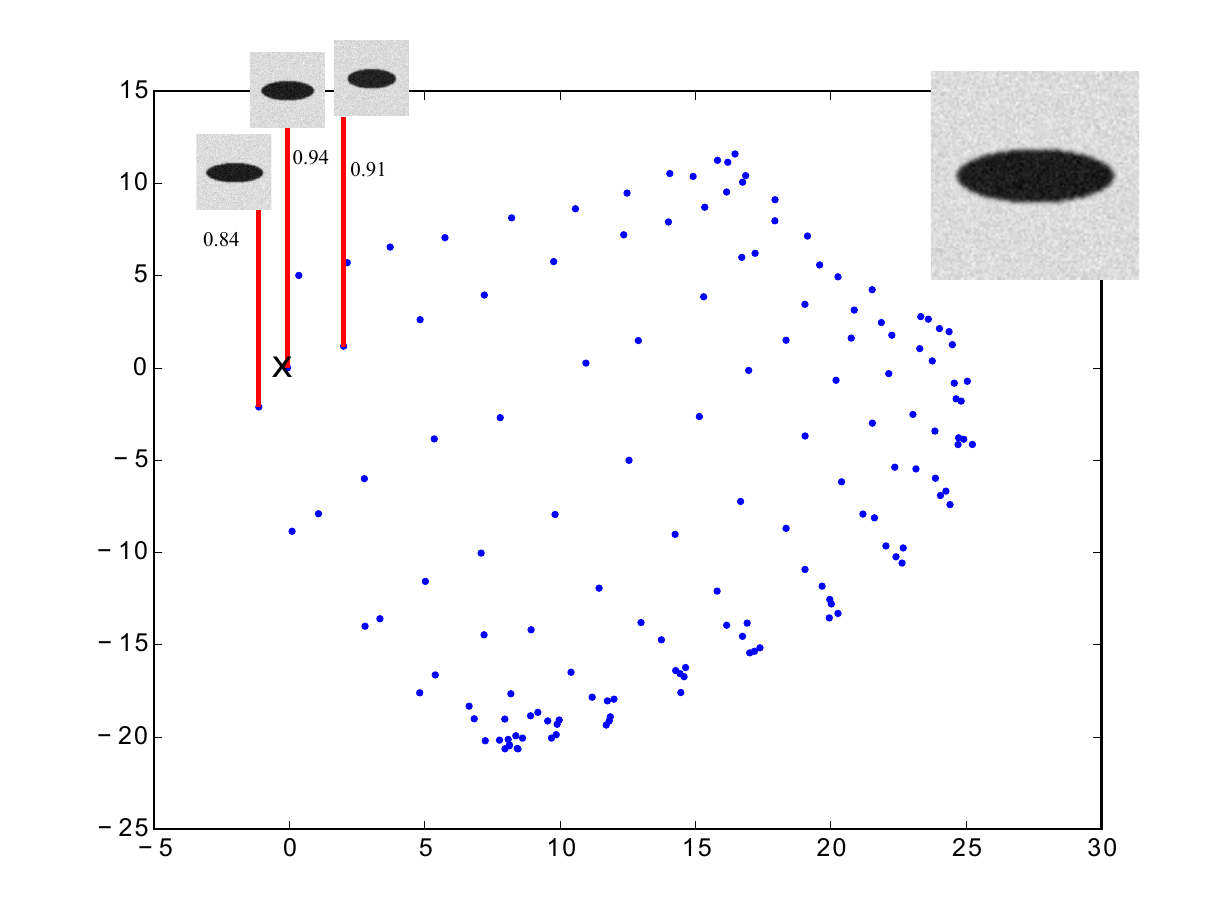}
\par\end{centering}
\caption{Mapping the ellipses $2$-dimensional manifold onto the coordinate system $H$. In the right upper corner we see the object we wish to project (i.e., $r$). Marked in $\times$ is the local origin $q$ and some nearby objects from the sampled data alongside their relative weights}
\label{fig:EllipsesH}\end{figure}

\subsection*{Comparison with PCA}
Although the main achievement of the MMLS method lies in the approximation of general nonlinear manifolds, we conducted a comparison between the MMLS and PCA. 
This comparison was carried as a baseline sanity check in order to verify that our methodology compares well with the most well known technique for linear manifolds.
We have examined two test cases. The first case deals with linear data (with and without noise) and the second deals with the relative simple nonlinear case of a sphere in $\RR^3$. 

\subsubsection*{Case 1 -- Linear setting}
\begin{itemize}
    \item Choose 3 random axes in $\RR^{50}$ (denoted by $u_1, u_2$ and $u_3$).
    \item Take $125=5^3$ uniformly distributed samples of a 3d linear subspace embedded in $\RR^{50}$ (the samples are denoted $\{p_i\}_{i=1}^{125}$).
    Explicitly the $u_1, u_2, u_3$ coordinates of each $p_i$ are distributed as $\UU(-0.5,0.5)$. We denote by $U$ the space spanned by $u_1, u_2, u_3$.
    \item Add Gaussian noise $\epsilon_i\in\RR^{50}$ distributed in each coordinate as $\sim \NN(\mu = 0, \sigma = 0.3)$ for each point. That is, the noisy sample set $\{r_i\}_{i=1}^{125}$ is defined by $r_i = p_i +\epsilon_i$.
    \item Compute the leading $d$ principal components of the data set (after subtracting the sample mean). We denote by $U_{\text{PCA}}$ the space spanned by these $d$ leading principal components.
    \item Measure the PCA error by
    \[
        E_{\text{PCA}} = \frac{1}{125}\sum_{i=1}^{125} dist(P_{{\text{PCA}}}(r_i), P_{U}(P_{{\text{PCA}}}(r_i)))
    ,\]
    where $P_{\text{PCA}}(x), P_{U}(x)$ denote the projections of a point $x$ onto $U_{\text{PCA}}$ and $U$ respectively.
    \item Compute the 1-degree MMLS approximation of all points (denoted by $P_{\text{MMLS}}(r_i)$ for all $i=1,...,125$).
    \item Measure the MMLS error by
    \[
        E_{\text{MMLS}} = \frac{1}{125}\sum_{i=1}^{125} dist(P_{{\text{MMLS}}}(r_i), P_{U}(P_{{\text{MMLS}}}(r_i)))
    .\]
    \item In order to obtain statistics, this experiment have been performed 50 times and we calculated the average errors and standard deviations.
\end{itemize}
The experiment's results yielded $E_{\text{PCA}} = 0.59508 \rpm 0.0391$ and $E_{\text{MMLS}} = 0.57849 \rpm 0.027$.
As can be expected, there is no real difference in the approximation error of the two approaches and the computation time of PCA ($0.00348$ sec on the average) was significantly faster than that of the MMLS ($5.263$ sec on the average). As a sanity check, we computed the errors and statistics for the clean case as well (i.e., using $\{p_i\}$ instead of $\{r_i\}$) and both methods yielded exact reconstructions.

\subsubsection*{Case 2 -- Sphere setting}
\begin{itemize}
    \item Sample at random 100 vectors in $\RR^3$.
    \item Take their $z$ coordinate and replace it with the absolute value and normalize the vectors (this way we narrow the sphere to be a semi-sphere on the positive part of the space -- see Fig. \ref{fig:SemiSphere}).
    \item Compute the leading $d$ principal components of the data set (after subtracting the sample mean).
    \item Measure the PCA mean squared error by
    \begin{equation}
        E_{\text{PCA}} = \frac{1}{100}\sum_{i=1}^{100}\norm{P_{\text{PCA}}(r_i) - r_i}^2
    \end{equation}
    \item Compute the 1-degree and 2-degree MMLS approximation of all points (denoted by $P_{\text{MMLS}_m}(r_i)$ for all $i=1,...,100$ and $m=1,2$ denote the polynomial degreek).
    \item Measure the MMLS error by
    \[
        E_{\text{MMLS}_m} = \frac{1}{100}\sum_{i=1}^{100} \norm{P_{{\text{MMLS}_m}}(r_i) - r_i}^2
    ,\]
    where $m=1,2$ are the degrees of the local polynomial approximation.
    \item In order to obtain statistics, this experiment have been performed 50 times and we calculated the average errors and standard deviations.
\end{itemize}

\begin{figure}[ht]
\begin{centering}
\includegraphics[width={1\linewidth}]{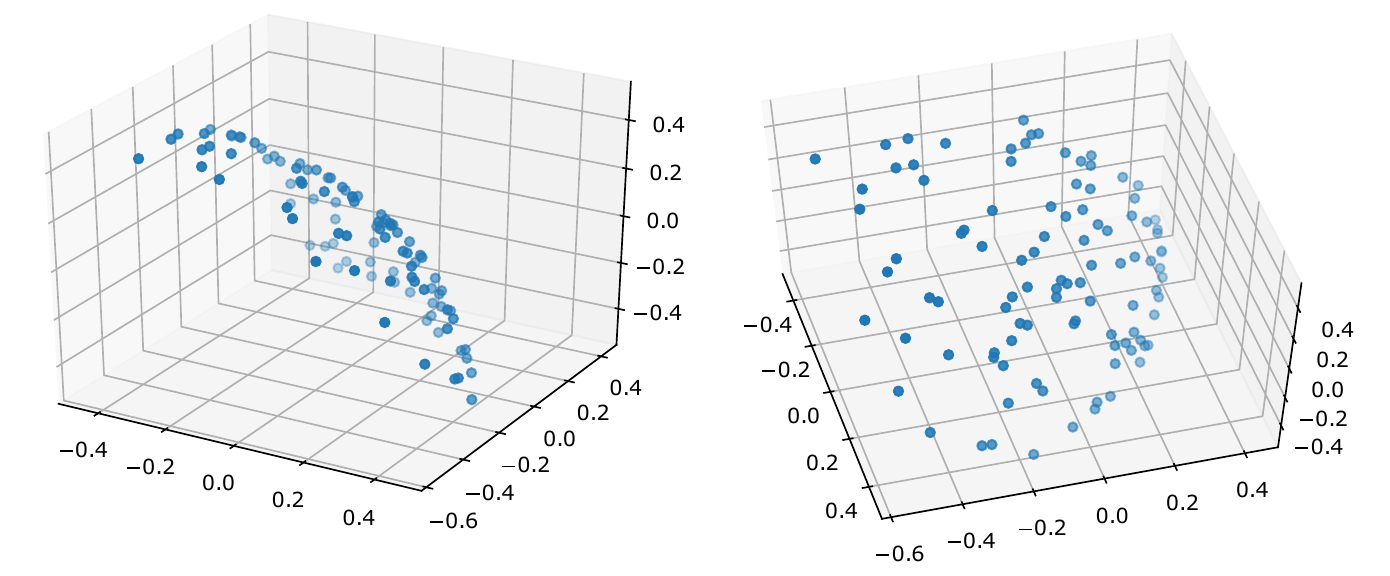}
\par\end{centering}
\caption{A specific sample of the semi-sphere after deducing the sample mean, as described in Case 2. The sample is shown from two different angles (left and right images). \label{fig:SemiSphere}}
\end{figure}

The experiment's results yielded $E_{\text{PCA}} = 0.78598 \rpm 0.00217$, $E_{\text{MMLS}_1} = 0.00337 \rpm 0.00079$ and $E_{\text{MMLS}_2} = 0 \rpm 0$ (up to machine precision).
Notice that the PCA reconstruction error is large, even though the data was sampled with no noise. This can be attributed to the non-linear nature of the manifold, which is not supposed to be well approximated by PCA. 
As can be seen, even the 1-degree MMLS improves significantly upon the PCA. It is not surprising that the error of the 2-degree MMLS yields zero error, as the sphere can be expressed as a second degree polynomial, locally.
The computation time of PCA ($0.00012$ sec on the average) was significantly faster than that of the MMLS ($2.42$ sec on the average).

\section*{Acknowledgements}
The authors wish to thank the referees as well as the journal's editor for their insightful remarks, which had an impact on the final version of paper.

\newpage
\subsection*{Appendix A - Geometrically Weighted PCA}
\label{sec:GeometricalPCA}
We wish to present here in our language the concept of \textit{geometrically weighted PCA} (presented a bit differently in \cite{harris2011geographicPCA}), as this concept plays an important role in some of the Lemmas proven in section \ref{sec:Theory} and even in the algorithm itself.

Given a set of $I$ vectors $x_1 ,..., x_I$ in $\mathbb{R}^n$, we look for a $Rank(d)$  projection $P \in \mathbb{R}^{n \times n}$ that minimizes:
\[
\sum\limits_{i=1}^I || Px_i - x_i ||_2^2
\]
If we denote by $A$ the matrix whose i'th column is $x_i$ then this is equivalent to minimizing:
\[
|| PA - A ||_F^2
,\] as the best possible $Rank(d)$ approximation to the matrix $A$ is the SVD $Rank(d)$ truncation denoted by $A_d$, we have:
\[
PA = P U \Sigma V^T = A_d = U \Sigma_d V^T
\]
\[
P = U \Sigma_d V^T V \Sigma^{-1} U^T
\]
\[
P = U \Sigma_d \Sigma^{-1} U^T
\]
\[
P = U I_d  U^T
\]
\[
P = U_d  U_d^T
\]
And this projection yields:
\begin{equation}
Px = U_d  U_d^T x = \sum\limits_{i=1}^d \langle x, u_i \rangle \cdot u_i
,\end{equation}
which is the orthogonal projection of $x$ onto $span \lbrace u_i \rbrace_{i=1}^d$. Here $u_i$ represents the i$^{th}$ column of the matrix $U$.

\begin{remark}
The projection $P$ is identically the projection induced by the PCA algorithm.
\end{remark}

\subsubsection*{The Weighted Projection:}
In this case, given a set of n vectors $x_1 ,..., x_I$ in $\mathbb{R}^n$, we look for a $Rank(d)$  projection $P \in \mathbb{R}^{n \times n}$ that minimizes:
\[
\sum\limits_{i=1}^I || Px_i - x_i ||_2^2 ~ \theta(||x_i - q||_2) =
\sum\limits_{i=1}^I || Px_i - x_i ||_2^2 ~ w_i
\] 
\[
= \sum\limits_{i=1}^I || \sqrt{w_i} Px_i - \sqrt{w_i} x_i ||_2^2 
\] 
\[
= \sum\limits_{i=1}^I || P  \sqrt{w_i} x_i - \sqrt{w_i} x_i ||_2^2 
\] 
\[
= \sum\limits_{i=1}^I || P y_i - y_i ||_2^2 
\] 
So if we define the matrix $\tilde{A}$ such that the i'th column of $\tilde{A}$ is the vector $y_i = \sqrt{w_i}x_i$ then we get the projection: 
\begin{equation}
P = \tilde{U}_d  \tilde{U}_d^T
\label{eq:geometricPCA}
,\end{equation}
where $\tilde{U}_d$ is the matrix containing the first $d$ principal components of the matrix $\tilde{A}$.

\bibliography{mybib}{}
\bibliographystyle{plain}

\end{document}